\documentclass[11pt]{article}
\usepackage{amsmath, amsthm, amssymb}
\usepackage{authblk}
\usepackage{sober}
\usepackage{graphicx}
\usepackage{color,soul}
\usepackage{nicefrac}
\usepackage{amsmath}
\usepackage{cite}
\usepackage{ragged2e}
\usepackage{algorithmic}
\usepackage{graphicx}
\usepackage{textcomp}
\usepackage{xcolor}
\usepackage{tcolorbox}
\usepackage{wrapfig}
\usepackage{bold-extra}
\usepackage{lipsum}
\usepackage{amsfonts}
\usepackage{graphicx}
\usepackage{stfloats}
\usepackage{enumitem}
\usepackage{xspace}
\usepackage{url}
\usepackage{comment}
\usepackage{epstopdf}
\usepackage{algorithmic}
\usepackage{algorithm2e}
\usepackage{tabularx}
\usepackage{setspace}
\usepackage{caption}
\usepackage{xcolor}
\usepackage{subcaption}
\usepackage{wrapfig}
\usepackage{array}
\usepackage{array}
\usepackage{tabulary}
\usepackage{ctable} 
\usepackage[normalem]{ulem} 
\usepackage{etoolbox}
\usepackage[normalem]{ulem}
\usepackage{adjustbox}
\usepackage[T1]{fontenc} 

\apptocmd{\sloppy}{\hbadness 10000\relax}{}{}
\usepackage{soul}
\makeatletter
\def\SOUL@hlpreamble{%
\setul{\dimexpr\dp\strutbox+1.5pt}{\dimexpr\ht\strutbox+\dp\strutbox+1.5pt\relax}
\let\SOUL@stcolor\SOUL@hlcolor
\SOUL@stpreamble
}

\newtheorem{innercustomthm}{Theorem}
\newenvironment{customthm}[1]
  {\renewcommand\theinnercustomthm{#1}\innercustomthm}
  {\endinnercustomthm}

\newtheorem{theorem}{Theorem}
\newtheorem{corollary}{Corollary}
\newtheorem{lemma}{Lemma}

\newtheorem{definition}{Definition}

\newif\ifcomments
\commentstrue

\newcommand{\AOneL}{c(\mbox{\Aone, \texttt{L}})\xspace}
\newcommand{\AOneH}{c(\mbox{\Aone, \texttt{U}})\xspace}
\newcommand{\ATwoL}{c(\mbox{\Atwo, \texttt{L}})\xspace}
\newcommand{\ATwoH}{c(\mbox{\Atwo, \texttt{U}})\xspace}

\newcommand{\Aone}{$\text{A}1$\xspace}
\newcommand{\Atwo}{$\text{A}2$\xspace}

\newcommand{\POW}{PoW\xspace}

\newcommand{\AlgA}{\textsc{SybilControl\xspace}}
\newcommand{\AlgB}{\textsc{CCom}\xspace}
\newcommand{\AlgC}{REMP\xspace}

\newcommand{\Diffuse}{\textsc{Diffuse}\xspace}

\newcommand{\Iters}{\mathcal{I}}

\newcommand{\cgoal}{Committee Invariant\xspace}
\newcommand{\sgoal}{Population Invariant\xspace}

\newcommand{\joinRate}{J^G\xspace}  
\newcommand{\iJRate}{J^G\xspace} 
\newcommand{\joinRateBad}{J^B\xspace}

\newcommand{\depRateTot}{D\xspace}  

\newcommand{\joinRateAll}{J^{\mbox{\tiny all}}\xspace}  

\newcommand{\elen}{\ell\xspace} 
\newcommand{\epochRate}{\rho\xspace}  

\newcommand{\estSet}{\tilde{S}\xspace} 
\newcommand{\estDur}{\tilde{\elen}\xspace}  

\newcommand{\advCost}{\mathcal{T}\xspace}
\newcommand{\advAveCost}{T\xspace}

\newcommand{\algGM}{\textsc{GMCom}\xspace}
\newcommand{\genID}{\textsc{GenID}\xspace}

\newcommand{\setIDs}{\mathcal{S}\xspace}
\newcommand{\setGood}{\mathcal{G}\xspace}

\newcommand{\iterIDs}{\mathcal{S}\xspace}
\newcommand{\curIDs}{\mathcal{S}_{\mbox{\tiny cur}}\xspace}

\newcommand{\JoinEst}{\tilde{J}^{G}\xspace}
\newcommand{\bigconstant}{c(JE, \texttt{H})\xspace}
\newcommand{\smallconstant}{c(JE, \texttt{L})\xspace}

\newcommand{\defn}[1]{\textbf{\emph{#1}}}
\newcommand{\EstGoodJoin}{\textsc{Estimate-GoodJR}\xspace}

\newcommand{\tog}{\textsc{ToGCom}\xspace}

\begin{document}
\title{\tog: An Asymmetric Sybil Defense
}
\date{}

\author[1]{Diksha Gupta}
\author[2]{Jared Saia\thanks{This work is supported by the National Science Foundation grants CNS-1318880 and CCF-1320994.}}
\author[3]{Maxwell Young\thanks{This work is supported by the National Science Foundation grant CCF 1613772 and by a research gift from C Spire.}}
\affil[1]{\small Dept. of Computer Science, University of New Mexico, NM, USA\hspace{6cm} \mbox{\texttt{dgupta@unm.edu}}}
\affil[2]{\small Dept. of Computer Science, University of New Mexico, NM, USA\hspace{6cm} \mbox{\texttt{saia@cs.unm.edu}}}
\affil[3]{\small Computer Science and Engineering Dept., Mississippi State University, MS, USA\hspace{6cm} \texttt{myoung@cse.msstate.edu}}

\maketitle              
\begin{abstract}
Proof-of-work (\POW) is one of the most common techniques to defend against Sybil attacks.  
Unfortunately, current \POW defenses have two main drawbacks.  First, they require  work to be done even in the absence of attack.  Second, during an attack, they require the good identities (IDs) to spend as much as the attacker.

\smallskip

Recent theoretical work by Gupta, Saia and Young suggests the possibility of overcoming these two drawbacks.  In particular, they describe a new algorithm, GMCom, that always ensures that a minority of IDs are Sybil.   They show that rate at which all good IDs perform computation is $O(\joinRate + \sqrt{T(\joinRate+1)})$, where {\boldmath{$\joinRate$}} is the join rate of good IDs, and {\boldmath{$T$}} is the rate at which the adversary performs computation. 


\smallskip

Unfortunately, this cost bound only holds in the case where (1) GMCom always knows the join rate of good IDs; and (2) there is a fixed constant amount of time that separates join events by good IDs.  Here, we present \textbf{\tog}, which removes these two shortcomings.  To do so, we design and analyze a  mechanism for estimating the join rate of good IDs; and also devise a new method for setting the computational cost to join the system. Additionally, we evaluate the performance of \tog alongside prior \POW-based defenses. Based on our experiments, we design heuristics that further improve the  performance of \tog by up to $3$ orders of magnitude over these previous Sybil defenses.


\end{abstract}

\section{Introduction}

A \defn{Sybil attack} occurs when a single adversary pretends to be multiple identities (\defn{IDs})~\cite{douceur02sybil}.  One of the oldest defenses against Sybil attacks is \defn{proof-of-work (\POW)} \cite{dwork:pricing}, in which any ID that wishes to use network resources or participate in group decision-making must first perform some work, typically solving computational puzzles.  

 \POW defends against any resource-bounded adversary. Thus, it is broadly applicable, in contrast to approaches relying on domain-specific attributes, such as social-network topology, multi-channel wireless communication, or device locality.  

Unfortunately, a significant drawback of \POW defenses is ``the work".  In particular, current \POW approaches have significant computational overhead, given that puzzles must always be solved,  even when the system is not under attack.  This non-stop resource burning translates into a substantial energy and, ultimately, a monetary cost~\cite{economistBC,coindesk,arstechnica}.  Consequently,  \POW approaches are used primarily in applications where participants have a financial incentive to continually perform work, such as Bitcoin~\cite{nakamoto:bitcoin}, and other blockchain technologies~\cite{litecoin,blockstack}.
 This is despite numerous proposals for PoW-based defenses in other domains~\cite{parno2007portcullis,wang:defending,kaiser:kapow,green:reconstructing,feng:design,waters:new,martinovic:wireless,
borisov:computational,li:sybilcontrol}.

Recently, Gupta et al.~\cite{Gupta_Saia_Young_2019} described an algorithm, {\bf \algGM}, that addresses this drawback. Let an ID be called \defn{bad} if it is controlled by the Sybil adversary, and \defn{good} otherwise. \algGM ensures that a minority of IDs are bad, and that good IDs spend in total asymptotically less than the adversary. In particular, define the \defn{spend rate} as the computational cost over  all good IDs per second, where this cost is due to solving puzzles. Then, \algGM  ensures the algorithm spend rate is $O(\joinRate + \sqrt{T(\joinRate+1)})$, where {\boldmath{$\joinRate$}} is the join rate of good IDs, and {\boldmath{$T$}} is the adversary's spend rate.  A lower-bound is given showing that this spend rate is asymptotically optimal~\cite{Gupta_Saia_Young_2019}.  This defense is called \defn{asymmetric} since the algorithm's spend rate is sublinear in $T$.   

Unfortunately, this spend rate only holds in the case where (1) \algGM always knows the join rate of good IDs; and (2) there is a fixed constant amount of time that separates all join events by good IDs. We illustrate the shortcoming of \algGM's (incorrect) join-rate estimate in Section~\ref{section:empasym}, and the effect of arbitrarily-close join events can have on \algGM in Appendix~\ref{app:GMComFailure}.

\subsection{Our Contributions}\label{sec:contributions}

We introduce a new algorithm \textbf{\textsc {\underline{T}}otal \underline{o}ver \underline{G}ood \underline{Com}putation} (\defn{\tog}).\footnote{So named since computational cost to join is proportional to total join rate over good join rate (See Section~\ref{sec:appendGMCom}).} Like \algGM, \tog maintains the following two invariants, which limit the Sybil adversary's ability to monopolize resources and control the network.

\begin{itemize}[leftmargin=12pt]
\item {\defn{\sgoal}{\bf:}} The fraction of bad IDs in the system is always less than $1/6$.
\item {\defn{\cgoal}{\bf:}} There is always a \defn{committee}, of size logarithmic in the current system size, known to all good IDs, that contains less than a $1/2$-fraction of bad IDs.\footnote{These constants can be can be made smaller, up to $< \alpha$, at a cost of increasing the hidden constants in our resource costs.}
\end{itemize}

\noindent Additionally, we make the following new contributions.

\begin{enumerate}[leftmargin=12pt]

\item \tog spends at a rate of $O(\joinRate + \sqrt{T(\joinRate+1)})$, even (1) without always knowing the join rate of good IDs; and (2) when there is no lower bound on the time between join events of good IDs.

\item We simplify and reduce from $4$ to $2$ the number of assumptions on the behavior of good IDs that are needed for our analysis.

\item We empirically compare \tog against prior \POW defenses using real-world data from several networks. \tog performs up to $2$ orders of magnitude better than previous defenses, according to our simulations (Section~\ref{section:empasym}).

\item Using insights from our first experiments, we engineer and evaluate several heuristics aimed at further improving the performance of \tog. Our best heuristic performs up to 3 orders of magnitude better than previous algorithms for large-scale attacks (Section~\ref{section:heuristics}).
\end{enumerate}

\subsection{The General Network Model}\label{sec:model-main}

We now describe a general network model that aligns with many permissionless systems, including the work in~\cite{Gupta_Saia_Young_2019}. 

 \smallskip

\noindent{\bf Puzzles.}\label{sec:puzzle} IDs can construct computational puzzles of varying hardness, whose solutions cannot be stolen or pre-computed. A {\boldmath{$k$}}\defn{-hard puzzle} for any integer $k \geq 1$ imposes a computational cost of $k$ on the puzzle solver. These are common assumptions in \POW systems~\cite{nakamoto:bitcoin, li:sybilcontrol, andrychowicz2015pow}.  

\smallskip

\noindent{\bf Communication.}\label{sec:com} All communication among good IDs uses a broadcast primitive,~{\bf \Diffuse}, which  allows  a  good  ID  to  send  a  value  to all other good IDs within a known and bounded amount of time, despite an adversary. Such a primitive is a standard assumption in \POW schemes~\cite{Garay2015,bitcoinwiki,GiladHMVZ17,Luu:2016}; see~\cite{miller:discovering} for empirical justification.

When a message is diffused in the network, it is not possible to determine which ID initiated the diffusion of that message.  Good IDs have digital signatures, and each message originating at a good ID is signed by its private key; we note that no public key infrastructure is assumed.

A \defn{round} is the amount of time it takes to solve a $1$-hard puzzle plus the (shorter) time to diffuse the solution to the rest of the network.\footnote{Communication latency in Bitcoin is 12 seconds; puzzle solving time is 10 minutes~\cite{croman2016scaling}}   All IDs are assumed to be synchronized. 

\smallskip

\noindent{\bf Adversary.}\label{sec:adv} A single adversary controls all bad IDs.  This pessimistically represents perfect collusion and coordination by the bad IDs. Bad IDs may arbitrarily deviate from our protocol, including sending incorrect or spurious messages. The adversary can send messages to any ID at will, and can read the messages diffused by good IDs before sending its own.  It knows when good IDs join and depart, but it does not know the private bits of any good ID.   

The adversary controls an $\alpha$-fraction of computational power, where $\alpha>0$ is a small constant.
  That is, in a single round where all IDs are solving puzzles, the adversary can solve an $\alpha$-fraction of the puzzles; this is common in past \POW literature~\cite{nakamoto:bitcoin, andrychowicz2015pow, walfish2010ddos, GiladHMVZ17}.

\smallskip

\noindent{\bf Joins and Departures.}\label{sec:join} At most a constant fraction of the good IDs join or depart in any round.  Departing good IDs announce their departure to the network. In practice, each good ID can issue ``heartbeat messages'' that are periodically diffused and indicate to the committee that this ID is still alive.

The minimum number of good IDs in the system at any point is assumed to be at least {\boldmath{$n_0$}}. The \defn{system lifetime} is defined to be the duration over which $n_0^{\gamma}$ joins and departures occur, for any fixed constant {\boldmath{$\gamma$}} $>0$.




\section{Our Algorithm - \tog }\label{sec:our-problems}


We begin by describing our algorithm, \tog, highlighting our new method for setting puzzle hardness. Next, we develop intuition for why this method yields the asymmetric property, and how it relies on a robust estimate of the good join rate. This motivates the use of \EstGoodJoin as another critical component in our algorithm design.

\begin{figure}[t!]
\centering
\begin{tcolorbox}[standard jigsaw, opacityback=0]
\begin{minipage}[h]{0.95\textwidth}
\noindent \textbf{\textsc{\underline{T}otal \underline{o}ver \underline{G}ood \underline{Com}putation (\tog)}}
\medskip
\footnotesize

\noindent{\bf Key Variables}\\
\smallskip
\begin{tabular}{p{0.6cm} p{0.1cm} p{9cm}}
	$i$ &:& iteration number\\
	$n_i^a$ &:& number of IDs joining since beginning of iteration $i$\\
	$n_i^d$&:& number of IDs departing since beginning of iteration $i$\\
	$\iterIDs_i$&:& set of IDs at end of iteration  $i$ \\
	$\curIDs$&:& current set of IDs in system
\end{tabular}
\smallskip

\noindent{\bf Initialization}\smallskip

\noindent $i \leftarrow 1$\\ 
$\setIDs_0 \leftarrow$ set of IDs returned by initialization phase\\ 
$\estSet \leftarrow \setIDs_0$\\
\noindent $\JoinEst_{0}  \leftarrow$ obtained in initialization phase
\medskip

\noindent{\bf Execution}\smallskip

\noindent The committee maintains all variables above and makes all decisions using Byzantine Consensus, including those used in \EstGoodJoin, which is run continuously. For each iteration $i$, do:
\begin{enumerate}[leftmargin=14pt] 
    \colorlet{light-gray}{gray!25}
    \sethlcolor{light-gray}
    \setlength{\lineskip}{0pt}
	\item[1.] Each joining ID solves and diffuses the solution to an \hl{entrance puzzle of difficulty equal to the number of IDs that have joined in the last $1/\scriptstyle{\JoinEst_{i}}$ seconds of the current iteration, including the newly joining ID.}\smallskip
	
	\item[2.] When  $n_i^a + n_i^d \geq (1/11)|\setIDs_{i-1}|$, do:\smallskip
	
			\textbf{Perform Purge}
			\begin{itemize}[leftmargin=15pt]
			\item[(a)] The committee generates and diffuses a random string $r$ to be used in puzzles for this purge and entrance for the next iteration. 
			\item[(b)] $\iterIDs_i$ $\leftarrow$  set of IDs returning difficulty $1$ puzzle solutions within $1$ round.
			\item[(c)] The committee selects a new committee of size $\Theta(\log n_0)$ from $\iterIDs_{i}$ and sends out this information via \Diffuse.
			\item[(d)] $i \leftarrow i+1.$
			\end{itemize}
            
            \smallskip

            \begin{tcolorbox}[left=0mm,top=0mm,bottom = 0mm,width = 0.93\textwidth,boxrule=0mm, arc= 0mm, colback = gray!25, colframe = white!10]
            {\textbf{Estimate-GoodJR}
            \begin{itemize}[leftmargin=15pt]
            \item[(e)]$\estSet \leftarrow$ \mbox{most recent membership                   ensuring}\mbox{$|\curIDs - \estSet | \geq  (3/5)|\curIDs|$.} 
            \item[(f)] $\estDur\leftarrow$ \mbox{length of time between last two                  changes}\mbox{ of variable $\estSet$.}
            \item[(g)]  $\JoinEst_i \leftarrow |\curIDs|/\estDur$
            \end{itemize}}
            \end{tcolorbox}
\end{enumerate}
\end{minipage}
\end{tcolorbox}
\caption{Pseudocode for \tog.}
\vspace{-5pt}
\label{alg:gmcom}
\vspace{-10pt}
\end{figure}

\subsection{Overview of \tog}\label{sec:appendGMCom}

We describe  \tog while referencing its pseudocode in Figure~\ref{alg:gmcom}.  \tog differs critically from \algGM in two ways: (1) the method for calculating the difficulty of puzzles assigned to joining IDs, and (2) the estimation of the good join rate.   The first change handles problems that may arise when good join events occur arbitrarily close together (see Section~\ref{app:GMComFailure}). The second enables a constant-factor estimate of the good join rate as shown in Section~\ref{sec:estimating}.
These two changes are highlighted in the pseudocode.

Execution occurs over disjoint periods of time called \defn{iterations}, and each iteration $i \geq 1$ consists of Steps 1 and  2. 

In Step 1, each joining ID must solve an \defn{entrance puzzle} of difficulty $1$ plus the number of IDs that join within the last $1/\JoinEst_i$ seconds where $\JoinEst_i$ is the estimate of the join rate of good IDs for iteration $i$ using \EstGoodJoin.  This entrance cost approximates the ratio of the total join rate over the good join rate, which motivates the name {\textsc{\underline{T}otal \underline{o}ver \underline{G}ood \underline{Com}putation}}.

Step 1 lasts until the earliest point in time when the number of IDs that join in iteration $i$, {\boldmath{$n_i^a$}},  plus the number of IDs that depart in iteration $i$,  {\boldmath{$n_i^d$}}, is at least  $(1/11)|\setIDs_{i-1}|$. The quantities $n_i^a$ and $n_i^d$ are tracked by the committee.

When Step 1 ends, a \defn{purge} is performed by the committee by issuing a  $1$-hard puzzle via \Diffuse in Step 2(a). In Step 2(b), each ID must respond with a valid solution within $1$ round. The committee removes unresponsive or late-responding IDs from its whitelist, which is maintained using Byzantine consensus amongst the committee members. 

The current committee then selects $\Theta(\log n_0)$ IDs uniformly at random from  $\setIDs_{i}$ in Step 2(c).  The committee uses \Diffuse to inform $\setIDs_{i}$ that the selected IDs are the new committee for iteration $i+1$. All messages from committee members are verified via public key digital signatures.  This only requires that all IDs know the digital signatures of the $\Theta(\log n_0)$ good committee members. This information is diffused in Step 2(c); we omit this detail in Figure~\ref{alg:gmcom} for ease of presentation. Since the committee has a good majority and coordinates its actions via Byzantine consensus, no public-key infrastructure is required.

The remainder of Step 2 consists of \EstGoodJoin, our new procedure for estimating the good-ID join rate. This estimate is computed in Step 2(g), and used for setting the entrance cost. In Section~\ref{sec:estimating}, we give a detailed description of \EstGoodJoin; its proof of correctness is given in Section~\ref{app:estimate}.

Finally, system initialization is achieved by solving the {\defn{\genID}} problem, where there is a set of good IDs, and an adversary with an $\alpha$-fraction of computational power who controls bad IDs. All good IDs must agree on a set of IDs that contains (1) all good IDs, and (2) at most an $\alpha$-fraction of bad IDs.  Solving \genID is a heavy-weight operation, but \algGM does this only once at system initialization. A number of algorithms exist for solving \genID~\cite{andrychowicz2015pow,hou2017randomized}.


\subsection{Developing Intuition for the Asymmetric Property}\label{sec:entrance-cost} 


Initially, the asymmetric result may be surprising, and so we offer intuition for this.  Consider iteration $i$.  In the absence of an attack, the entrance cost should be proportional to the good join rate. This is indeed the case since the puzzle difficulty is $O(1)$ corresponding to the number of (good) IDs that join within the last $1/\JoinEst_i  \approx 1/\iJRate_i$ seconds.

In contrast, if there is a large attack, then the entrance-cost function imposes a significant cost on the adversary. Consider the case where a batch of many bad IDs is rapidly injected into the system. This drives up the entrance cost since the number of IDs joining within  $1/\JoinEst_i$ seconds increases.

More precisely, assume the adversary's spending rate is $T = \xi  \joinRateAll_i$, where {\boldmath{$\xi$}} is the entrance cost, and {\boldmath{$\joinRateAll_i$}} is the join rate for all IDs. For the good IDs, the spending rate due to the entrance cost is $\xi \joinRate_i$, and the spending rate due to the purge cost is $\joinRateAll_i$. Setting these to be equal, and solving for $\xi$, we get $\xi  = \joinRateAll_i/ \joinRate_i$;  in other words, the number of IDs that have joined over the last $1/ \joinRate_i$ seconds.  This is the entrance cost function that best balances entrance and purge costs.

Spending rate of good IDs due to the entrance costs and purge costs is:
$$\xi \joinRate_i + \joinRateAll_i  \leq  2\joinRateAll_i = 2\sqrt{\left(\joinRateAll_i\right)^2} = 2 \sqrt{\joinRateAll_i  \xi \joinRate_i} =  2\sqrt{\joinRate_iT},$$
where the first inequality holds by our setting of $\xi$, the third step since $\joinRateAll_i = \xi \joinRate_i$, and the final step since $T = \xi  \joinRateAll_i$. This informal analysis shows how knowledge of the good join rate can be used to reduce the algorithmic spend rate.

\subsection{\EstGoodJoin}\label{sec:estimating}

Given the above intuition, a method for estimating $\joinRate_i$ is needed. However, good IDs cannot be discerned from bad IDs upon entering the system,  and so the adversary may inject bad IDs in an attempt to obscure the true join rate of good IDs. Designing a robust procedure for obtaining   $\joinRate_i$ is tricky. In this section, we describe our estimation algorithm, {\textbf{\EstGoodJoin}}, defined in Figure~\ref{alg:gmcom}, Steps (e) - (g), and later prove its correctness in Section~\ref{sec:theoretical}.

For any dynamic system, our analysis makes use of a disjoint period of time called an \defn{epoch}; roughly, this is the duration of time until the system membership of good IDs changes by a constant fraction.

\begin{definition}\label{def:epoch}
Let {\boldmath{$\setGood_{i}$}} be the set of good IDs in the system at the end of epoch $i$ and let epoch $1$ begin at system initialization. Then, epoch $i$ is defined as the shortest amount of time until $|\setGood_{i} - \setGood_{i-1}| \geq (3/4)|\setGood_{i}|$.
\end{definition}

Let {\boldmath{$\epochRate_j$}} be the join rate of good IDs in epoch $j$; that is, the number of good IDs that join in epoch $j$ divided by the number of seconds in epoch $j$.  

We define two assumptions, \Aone and \Atwo, on the join rate of good IDs (not bad IDs). We let $c(\texttt{A}, \texttt{L})$ and  $c(\texttt{A}, \texttt{U})$ be positive constants used to, respectively, lower bound and upper bound a quantity pertaining to the assumption \texttt{A}. 
\vspace{-5pt}
\begin{itemize}[leftmargin=10pt]
\item \defn{Assumption {\Aone.}} For all $j > 1$, $\AOneL\, \epochRate_{j-1}  \leq \epochRate_{j} \leq \AOneH\, \epochRate_{j-1}$.\smallskip
\item \defn{Assumption {\Atwo.}} For any period of time within epoch $j$ that contains at least $2$ good join events, the good join rate during that period is between $\ATwoL\,\epochRate_{j}$ and $\ATwoH\, \epochRate_{j}$.
\end{itemize} 
\vspace{-5pt}
Informally, assumption \Aone implies that the {rate at which good IDs join does not change by too much from one epoch to the next}.  assumption \Atwo implies that two or more consecutive good join events cannot be too close or too spread out. In Section~\ref{s:JandLAssum},  we give empirical evidence supporting these two assumptions.


\subsection{Description of \EstGoodJoin}

\EstGoodJoin continually keeps track of the sets {\boldmath{$\estSet$}}, which is the most recent system membership such that  $|\curIDs - \estSet| \geq \frac{3}{5} |\curIDs|$, and {\boldmath{$\curIDs$}} is the set of  IDs currently in the system. Whenever $\estSet$ is updated --- and only when $\estSet$ is updated --- the parameter {\boldmath{$\estDur$}} is set to the length of time since the system membership was most recently $\estSet$, and we refer to this length of time as an \defn{interval}. Then, {\boldmath{$\JoinEst_i$}}$= |\curIDs|/\estDur$ is used as an estimate for the true good join rate over iteration $i$, denoted by {\boldmath{$\iJRate_i$}}.  

\begin{figure}[t!]
\centering 
\vspace{-10pt}
\includegraphics[ height=3.5cm]{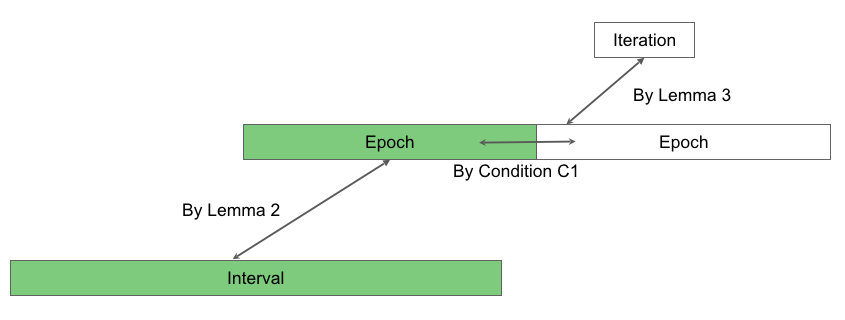}
\vspace{-4pt}\caption{\small A depiction of the relationship between the good join rate in intervals, epochs, and iterations. Arrows represent that  good join rates are within constant factors of each other. Green indicates the most recently-finished epoch and  interval being tracked by the committee to obtain $\JoinEst_i$.}
  \label{fig:estimating-JG}  
  \vspace{-10pt}
\end{figure}

Our approach for estimating the good-ID join rate  may be of independent interest for other settings where Assumptions \Aone and \Atwo hold. Therefore, we provide intuition for why they are necessary. 

 Since good and bad IDs cannot be distinguished with certainty,  it is challenging to determine when an epoch begins and ends.  Instead, \EstGoodJoin calculates a lower bound on the number of new good IDs that have joined the system by pessimistically subtracting out the fraction of IDs that could be bad.  When this estimate of new good IDs is sufficiently large --- that is, when the current interval ends ---  \EstGoodJoin ``guesses'' that at least one  epoch has occurred.  However, multiple (but still a constant number of) epochs may have actually occurred during this time. Thus, for this approach to yield a constant-factor estimate of $\iJRate_i$, the true good join rate cannot have changed ``too much" between these epochs, and \Aone bounds the amount of such change.

Assumption \Aone addresses the good join rate measured over entire, consecutive epochs. But, an interval may end {\it within} some epoch $j$. If the good join rate over the portion of epoch $j$ overlapped by the interval deviates by more than a constant factor from $\epochRate_j$, then the estimate will be inaccurate.  Assumption \Atwo ensures that this cannot happen.



\section{Theoretical Results}\label{sec:theoretical}

\smallskip

First, we prove that \EstGoodJoin achieves a constant-factor approximation, $\JoinEst_i$, to the good-ID join rate $\iJRate_i$ (Theorem \ref{t:JoinEst}). Define {\defn{``with high probability'' (w.h.p.)}} to mean with probability at least $1 - 1/n_0^{\gamma}$. All of our arguments hold with high probability over the system lifetime, and they rely on Assumptions \Aone and \Atwo. Due to space constraints, we include proofs in Appendix\ref{sec:appendProofCost}.

\begin{theorem} \label{t:JoinEst}
For any iteration $i\geq 1$, \EstGoodJoin provides a constant-factor estimate of  the good-ID join rate: 
$$ \smallconstant \iJRate_{i} \leq \hspace{-1pt}\JoinEst_i \hspace{-3pt} \leq \bigconstant \iJRate_{i}$$
where $\smallconstant = \left(\frac{5}{6}\right) \cfrac{\AOneL^2 \ATwoL}{\AOneH} $ and $\bigconstant = \cfrac{5\,\AOneH^2 \ATwoH}{\AOneL}$.
\end{theorem}

\noindent The following theorem bounds the spend rate of \tog.  We explicitly include the constants from Assumptions \Aone and \Atwo.
\begin{theorem}\label{thm:new-main-upper}
For $\alpha \leq 1/18$, w.h.p. over the system lifetime, \tog maintains the Population and Committee invariants and ensures:
$$\mathcal{A}\leq 11 d_2  \left(d_1\sqrt{2 \advAveCost(\bigconstant \joinRate + 1)}  + \joinRate \right)$$
\noindent where $d_1 =\sqrt{2\bigconstant}$ and $d_2 = \left( \frac{12}{11} + \dfrac{\AOneH \ATwoH}{11\smallconstant} \right)$.
\end{theorem}


The proof of Theorem~\ref{thm:new-main-upper} differs considerably from that of \algGM~\cite{Gupta_Saia_Young_2019} (see Appendix \ref{sec:appendasym}). We omit our arguments for the Population and Committee invariants, since those are unchanged from~\cite{Gupta_Saia_Young_2019}.


\section{Experiments}\label{sec:experiments} 

We now report our empirical results. In Section \ref{s:JandLAssum}, we test assumptions \Aone and \Atwo from Section~\ref{sec:estimating}. In section \ref{section:empasym}, we measure the computational cost for \tog, as a function of the adversarial cost, and compare it against prior PoW based algorithms.  Finally,  in Section \ref{section:heuristics}, we propose and implement several heuristics to improve the performance of \tog. All our experiments were written in MATLAB. 

\smallskip
\noindent{\bf Data Sets.} Our experiments use data from the following networks: \vspace{-7pt}
\begin{itemize}
\item[1.]  {\it Bitcoin:} This dataset records the join and departure events of IDs in the Bitcoin network, timestamped to the second, over roughly 7 days \cite{7140490}.  
\item[2.] {\it BitTorrent RedHat:} This dataset simulates the join and departure events for the BitTorrent network to obtain a RedHat ISO image. We use the Weibull distribution with shape and scale parameters of 0.59 and 41.0, respectively, from \cite{Stutzbach:2006:UCP:1177080.1177105}.
\item[3.] {\it Ethereum:} This dataset simulates join and departure events of IDs for the Ethereum network. Based on a study in \cite{kim2017measuring}, we use the Weibull distribution with shape parameter of 0.52 and scale parameter of 9.8. 
\item[4.] \textit{Gnutella:} This dataset simulates join and departure events for the Gnutella network. Based on a study in~\cite{rowaihy2007limiting}, we use an exponential distribution with mean of $2.3$ hours for session time, and Poisson distribution with mean of $1$ ID per second for the arrival rate. 
\end{itemize}

\subsection{Testing Assumptions \Aone and \Atwo}\label{s:JandLAssum}

\smallskip
\noindent{\bf Experimental Setup.} For the Bitcoin network, the system starts with $9212$ IDs, and the join and departure events are based on the dataset from \cite{neudecker-atc16}. For the other networks, we initialize the system with $1000$ IDs, and simulate the join and departure events over $1000$ epochs.  We assume all joining IDs are good.  Every value plotted is the mean of $20$ independent runs. 

To test Assumption \Aone, for each epoch $i \geq 2$, the good join rate in epoch $i$, $\rho_i$, is compared to the good join rate in the previous epoch, $\rho_{i-1}$. Results are summarized in the columns labeled $\AOneL$ and $\AOneH$ of Table \ref{tab:assumptions}.

To test Assumption \Atwo, for each epoch $i\geq 1$, we consider all periods of time containing at least $2$ good joins events. In particular, we measure the minimum and maximum join rate for epoch $i$, denoted by {\boldmath{$\rho_i^{min}$}} and {\boldmath{$\rho_i^{max}$}}, respectively, and compare these values against $\rho_i$. The results are presented in Table~\ref{tab:assumptions}.  We have included the corresponding plots in Appendix \ref{app:plots}.

\begin{table*}[h]
\centering
\begin{tabular}{|p{3.5cm}|p{1.5cm}|p{1.5cm}|p{1.5cm}|p{1.5cm}|}
\hline
\textbf{Network} & $\AOneL$ & $\AOneH$ & $\ATwoL$ & $\ATwoH$ \\
\hline
Bitcoin  & 0.1 & 10 & 0.0005 & 30 \\
\hline
BitTorrent RedHat & 0.125 & 8 & 0.067 & 15 \\
\hline
Ethereum Mainnet & 0.5 & 2 & 0.4 & 2 \\
\hline
Gnutella & 0.5 & 2 & 0.1 & 4 \\
\hline
\end{tabular}	
\caption{Constants Assumptions \Aone and \Atwo for Section \ref{s:JandLAssum} }
\label{tab:assumptions}
\end{table*}


 \subsection{Evaluating Computational Cost without Heuristics}\label{section:empasym}

We now measure the spend rate for \tog, focusing solely on the computational cost of solving puzzles.  Throughout, we assume a computational cost of $k$ for solving a puzzle of difficulty $k$. We compare the performance of \tog against four PoW-based Sybil defense algorithms:   \textbf{\algGM}~\cite{Gupta_Saia_Young_2019}, \textbf{\AlgB}~\cite{pow-without}, \textbf{\AlgA} \cite{li:sybilcontrol} and \textbf{\AlgC} (a name that uses the authors' initials)\cite{rowaihy2005limiting}, summarized below.

\smallskip
\noindent\textbf{\algGM.} \algGM is like \tog, except for two differences.  First, the entrance cost is the maximum of $1$, and the measured join rate in the current iteration divided by an estimate of the good join rate.  Second, the estimate of the good join rate is computed via a different (incorrect) heuristic~\cite{Gupta_Saia_Young_2019}.

\noindent\textbf{\AlgB.} \AlgB is the same as \tog except the entrance cost is always 1. \smallskip

\noindent\textbf{\AlgA.} Each ID solves a puzzle to join. Additionally, each ID tests its neighbors with a puzzle every $5$ seconds, removing from its list of neighbors those IDs that fail to provide a solution within a fixed time period.  These tests are not coordinated between IDs.\smallskip

\noindent\textbf{\AlgC.}  Each ID solves a puzzle to join. Additionally, each ID must solve puzzles every {\boldmath{$W$}} seconds. We use Equation (4) from \cite{rowaihy2005limiting} to compute the value of computational spending rate per ID as: 

\begin{equation}\label{eq:cal_w}
	 \frac{L}{W} = \frac{n}{ N_{attacker}} = \frac{T_{max}}{\alpha N}\\
\end{equation}

\noindent where $L$ is the computational cost to an ID per $W$ seconds, $n$ is the number of IDs that the adversary can add to the system and $N_{attacker}$ is the total number of attackers in the system. Suppose $N$ is the system size, then $N_{attacker}$ is $\alpha N$ since the computational power with the adversary is an $\alpha$  fraction of computational power of the network in our model. Suppose $T_{max}$ is the maximum number of attackers (bad IDs) that can join in $W$ seconds, then to guarantee that the fraction of bad IDs is less than half, $n = T_{max}$. Substituting these values in Equation \ref{eq:cal_w}, we can compute total algorithmic spending rate as:
\begin{equation}\label{eq:REMP}
	 \mathcal{A}_{REMP} = (1-\alpha)N \times \frac{L}{W} = \frac{(1-\alpha)T_{max}}{\alpha}
\end{equation}

\smallskip
\noindent{\bf Setup.}
We assume the good IDs join and depart as described in Section \ref{s:JandLAssum}. We set $\alpha = 1/18$, and  let $T$ range over $[2^0,2^{30}]$, where for each  value of $T$, the system is simulated for $10,000$ seconds. We also simulate the case $T = 0$. We assume that the adversary only solves puzzles to add IDs to the system.  For REMP, we consider two values of $T_{max}$, $10^4$ and $10^7$. Setting $T_{max} = 10^7$ ensures correctness for all values of $T$ considered, and $T_{max} = 10^4$ ensures correctness for $T \leq 10^4$. 

\smallskip
\noindent{\bf Results.} Figure \ref{fig:AvsT} 
illustrates our results; we omit error bars since they are negligible. The x-axis is the adversarial spending rate, $T$; and the y-axis is the algorithmic spending rate, $\mathcal{A}$.  We cut off the plots for REMP-$10^4$ and \AlgA, when they can no longer ensure that the fraction of bad IDs is less than $1/2$. We also note that REMP-$10^7$ only ensures a minority of bad IDs for up to $T = 10^7$. 

\tog always has spend rate as low as the other algorithms for $T \geq 100$, and significantly less than the other algorithms for large $T$, with improvements that grow to about $2$ orders of magnitude.  In Section~\ref{section:heuristics}, our heuristics close this gap, allowing \tog to outperform all algorithms for all $T \geq 0$.  The spend rate for \tog is linear in $\sqrt{T}$, agreeing  with our analytical results.  We emphasize that the benefits of \tog are consistent over four disparate networks.  These results illustrate the value of \EstGoodJoin.

Finally, we note that \tog guarantees a fraction of bad IDs no more than $1/6$ for all values of $T$. In contrast, \AlgA~and REMP guarantee a fraction of bad IDs less than $1/2$ for the values of $T$ plotted.

For $T\geq 100$, \algGM and \AlgB perform almost identically. This occurs because of an error in the estimation heuristic of~\cite{Gupta_Saia_Young_2019}, which causes \algGM to incorrectly estimate $\joinRate$, when $T$ is much larger than $\joinRate$ ($\joinRate \approx 10$ in these plots).  When the estimate is incorrect, by the specification in~\cite{Gupta_Saia_Young_2019}, the entrance-puzzle hardness is set to $1$, and so \algGM reverts to \AlgB.
Our simulations in~\cite{Gupta_Saia_Young_2019} assumed knowledge of the good join rate, and thus did not reveal the flaw in  \algGM's estimation method.


\subsection{Heuristics}\label{section:heuristics}

\begin{figure*}[t!]
\centering
\captionsetup[subfigure]{labelformat=empty}
\begin{subfigure}{0.42\textwidth}
	\includegraphics[trim = 1.8cm 8cm 1.8cm 8cm, width=1\textwidth]{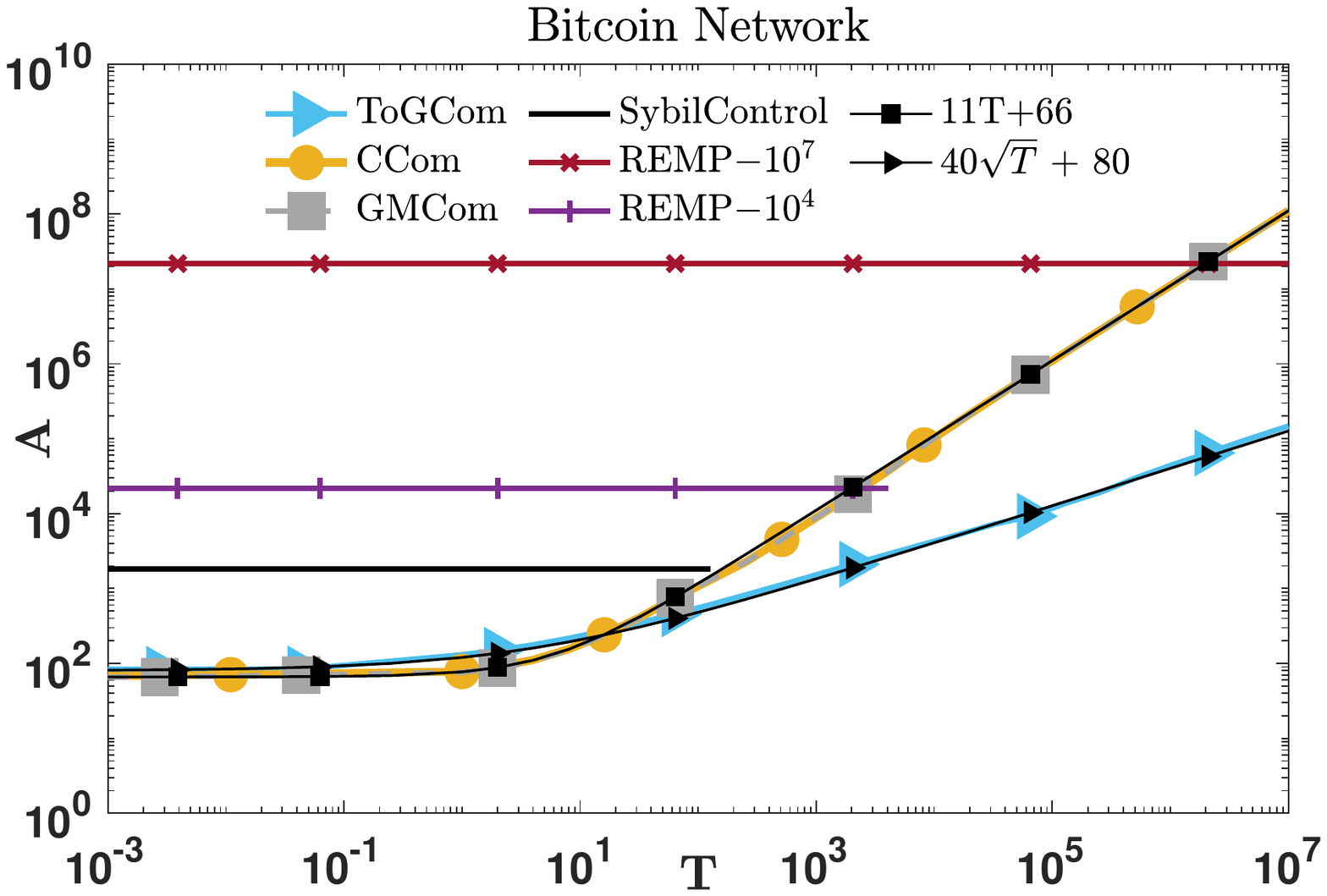}
	\vspace{-15pt}
\end{subfigure}
\vspace{2pt}
\begin{subfigure}{0.42\textwidth}
	\includegraphics[trim = 1.8cm 8cm 1.8cm 8cm, width=1\textwidth]{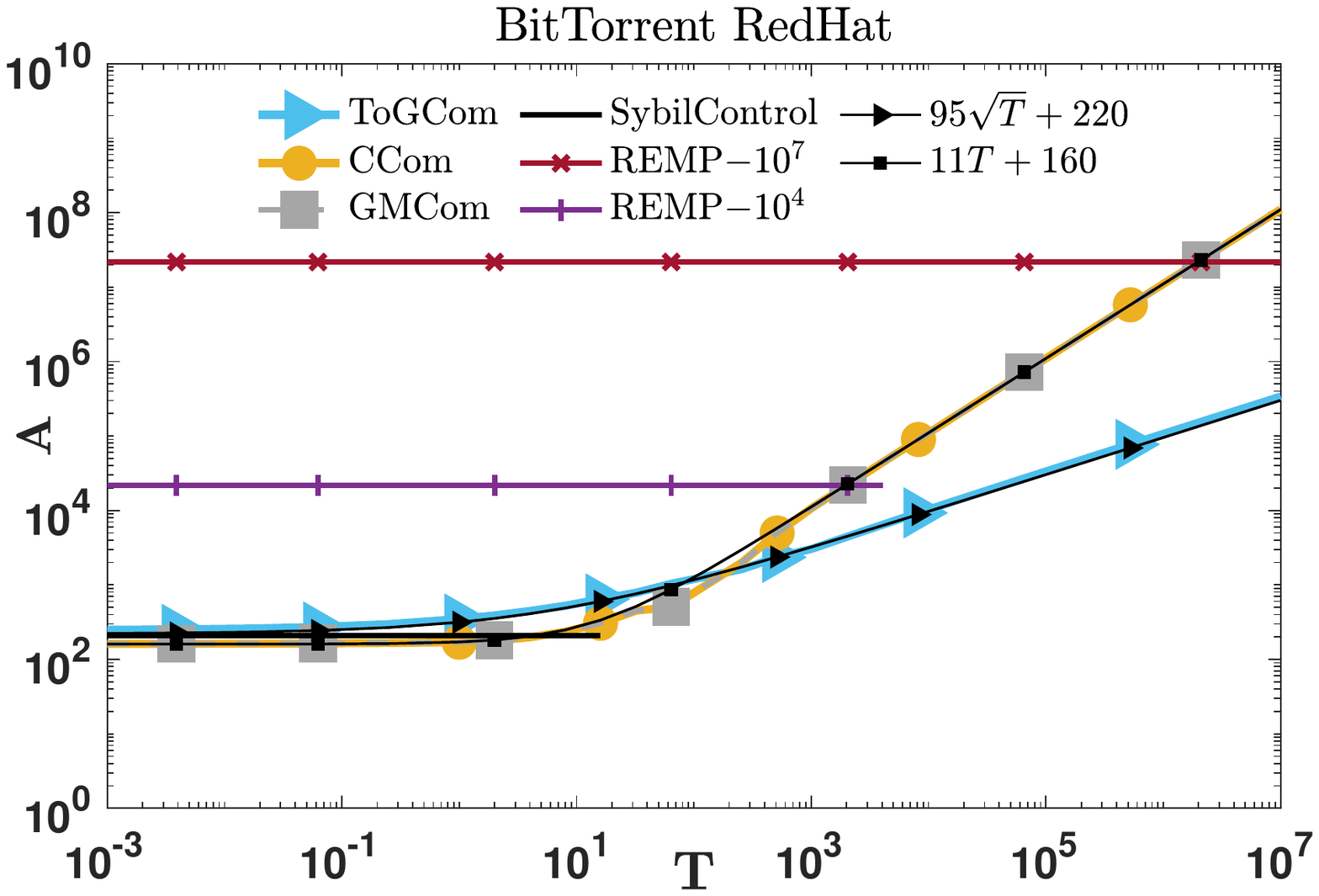}
	\vspace{-15pt}
\end{subfigure}
\vspace{2pt}
\begin{subfigure}{0.42\textwidth}
	\includegraphics[trim = 1.8cm 8cm 1.8cm 8cm, width=1\textwidth]{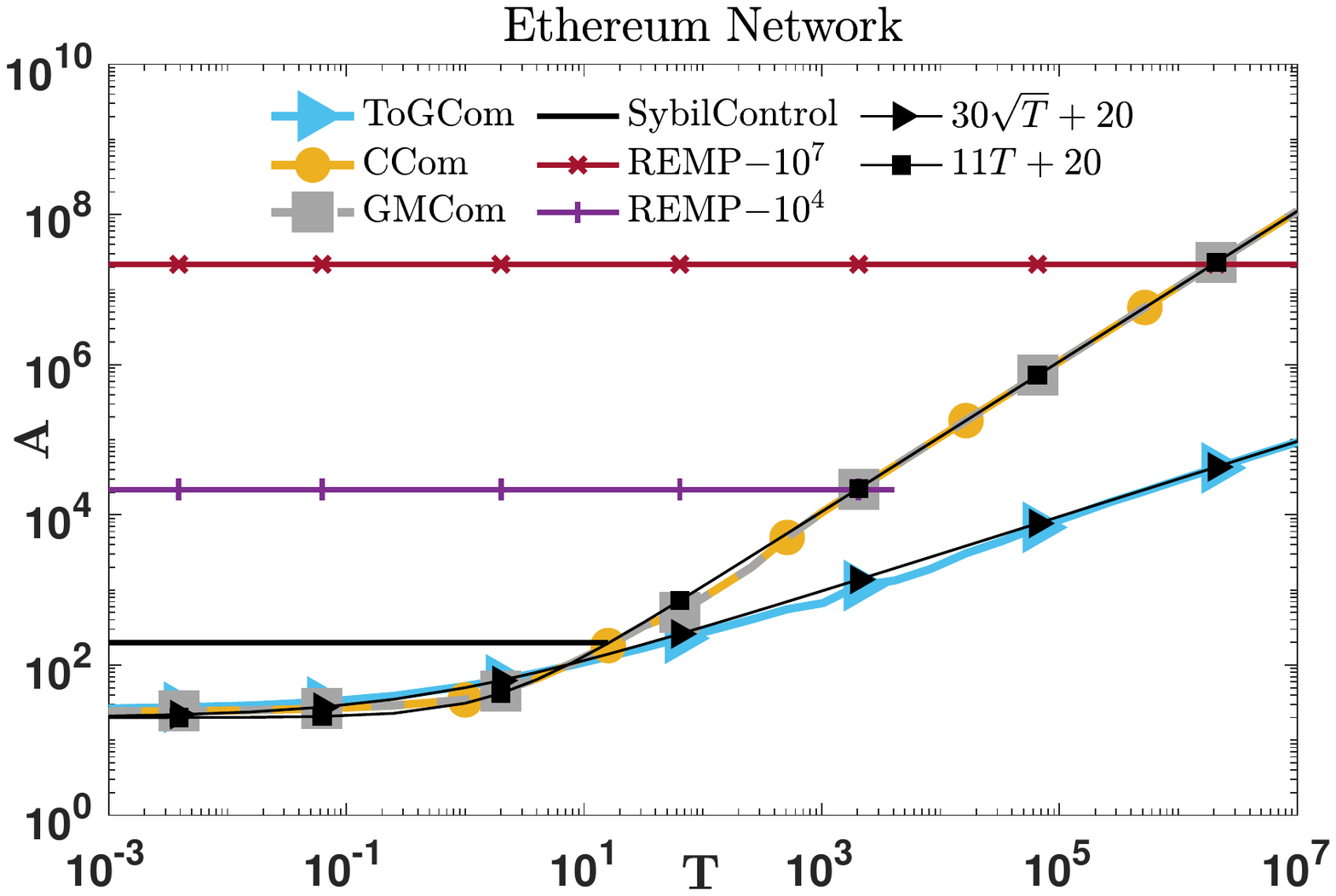}
	\vspace{-15pt}
\end{subfigure}
\begin{subfigure}{0.42\textwidth}
 	\includegraphics[trim = 1.8cm 8cm 1.8cm 8cm, width=1\textwidth]{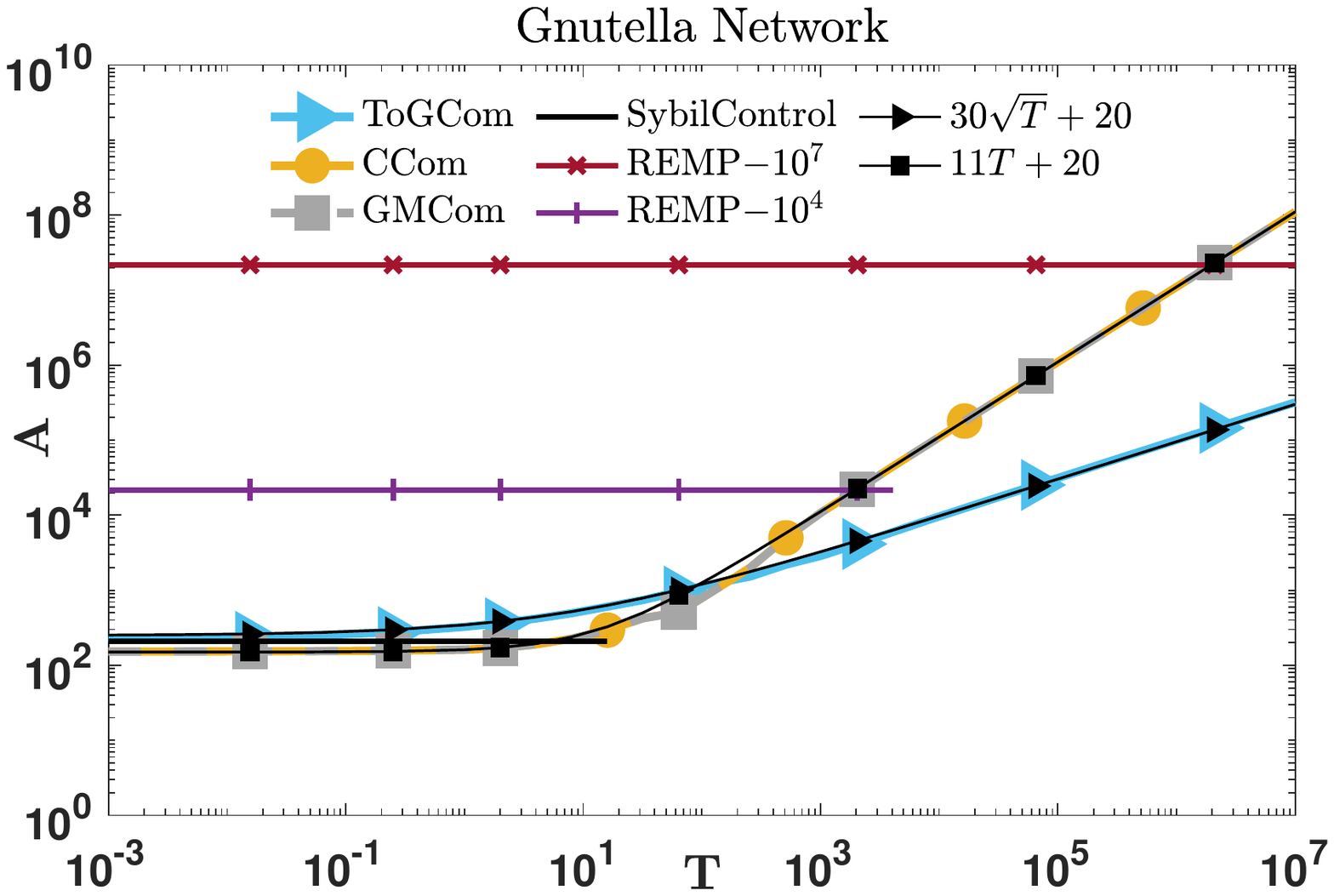}
	\vspace{-15pt}
\end{subfigure}
\vspace{-5pt}
\caption{Algorithmic cost versus adversarial cost for \tog, \algGM, \AlgB, SybilControl and REMP.}
\label{fig:AvsT}

\end{figure*} 

\smallskip

\begin{figure*}[t!]
\centering
\captionsetup[subfigure]{labelformat=empty}
\begin{subfigure}{0.42\textwidth}
	\centering
	\includegraphics[trim = 1.6cm 7.5cm 2cm 9cm, width=1\textwidth]{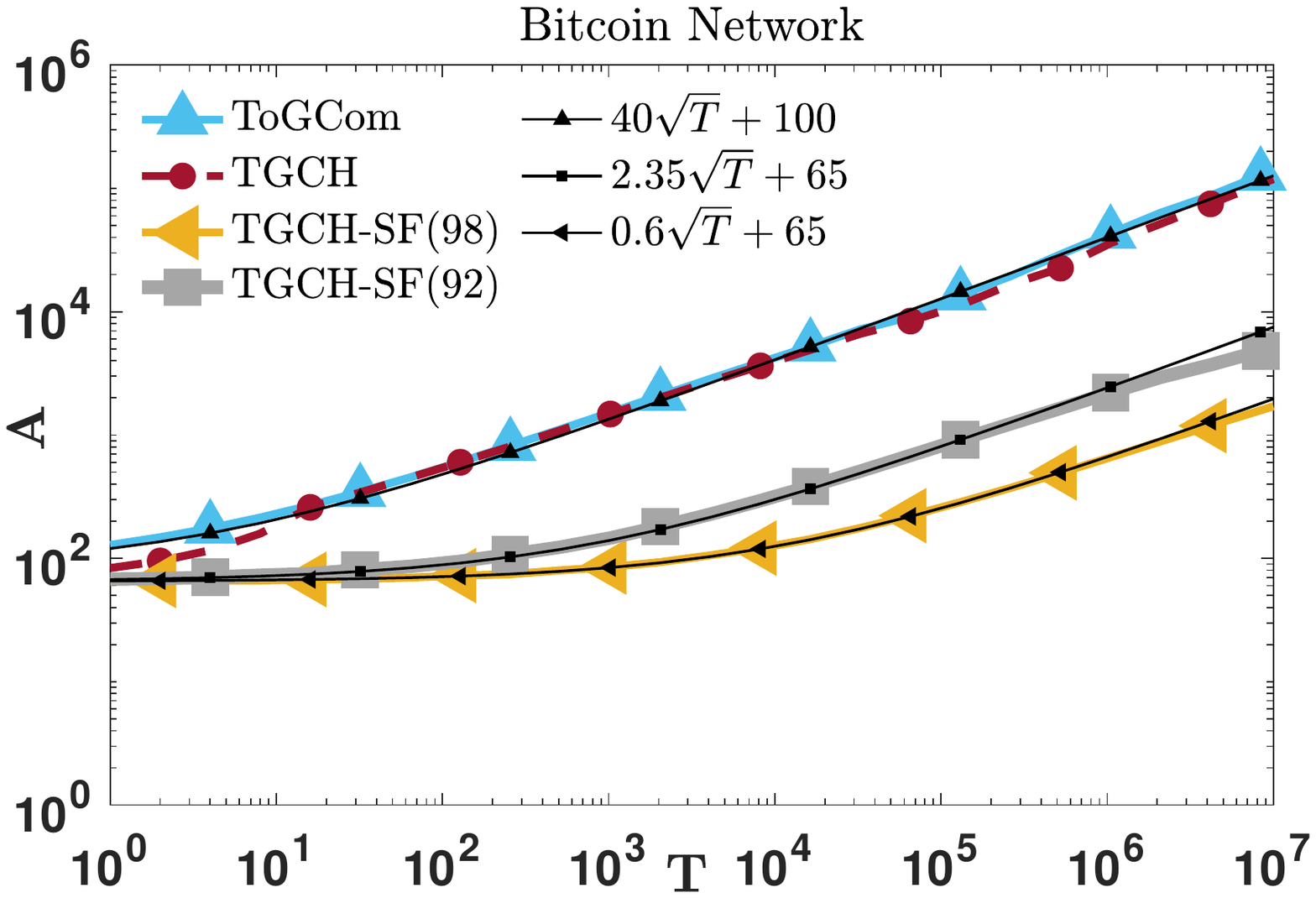} 
	\vspace{-20pt}
	\vspace{5pt}
\end{subfigure}
\begin{subfigure}{0.42\textwidth}
	\includegraphics[ trim = 1.6cm 7.5cm 2cm 9cm, 
	width=1\textwidth] {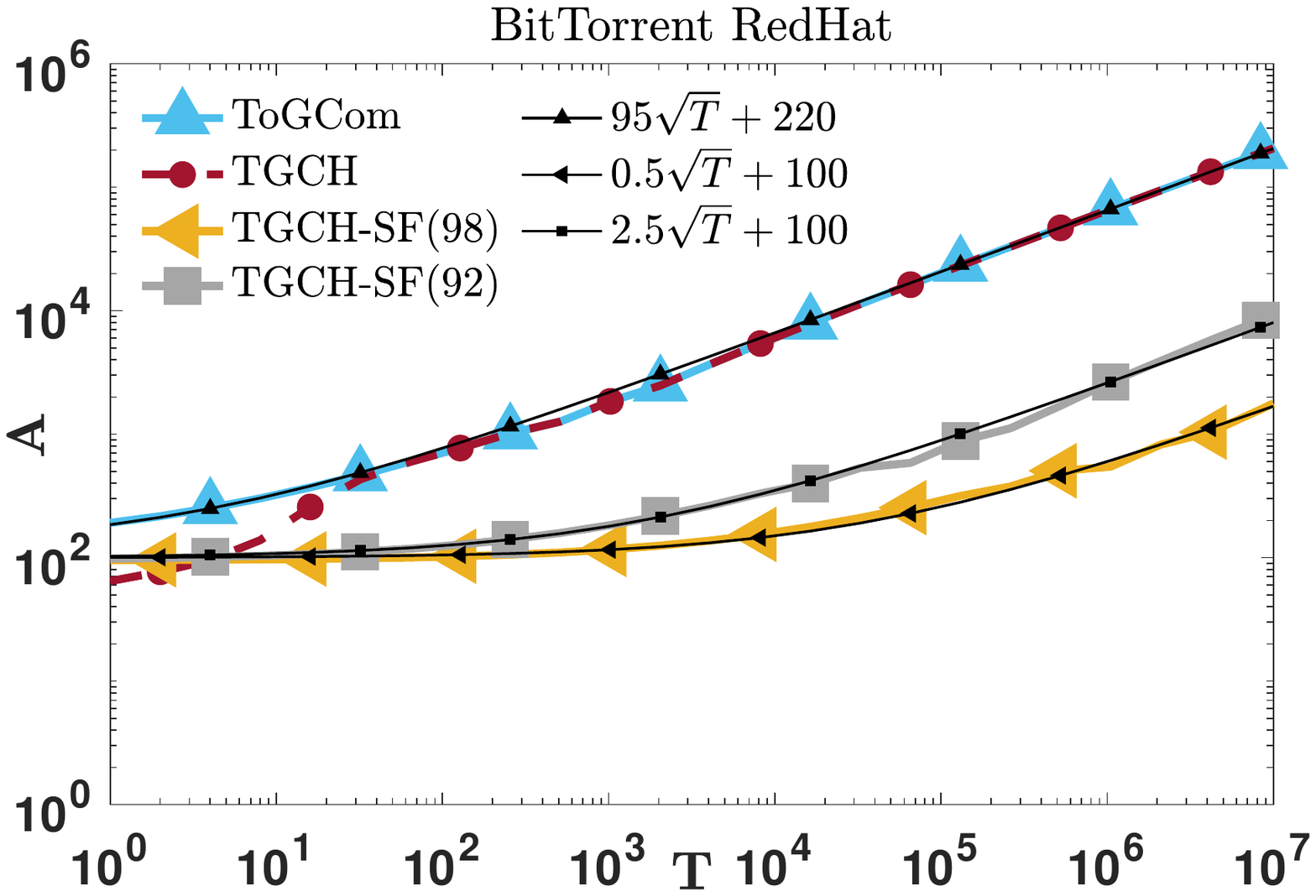}
	\vspace{-20pt}
	\vspace{5pt}
\end{subfigure}
\begin{subfigure}{0.42\textwidth}
	\includegraphics[trim = 1.5cm 8cm 1.5cm 8cm,   width=1\textwidth]{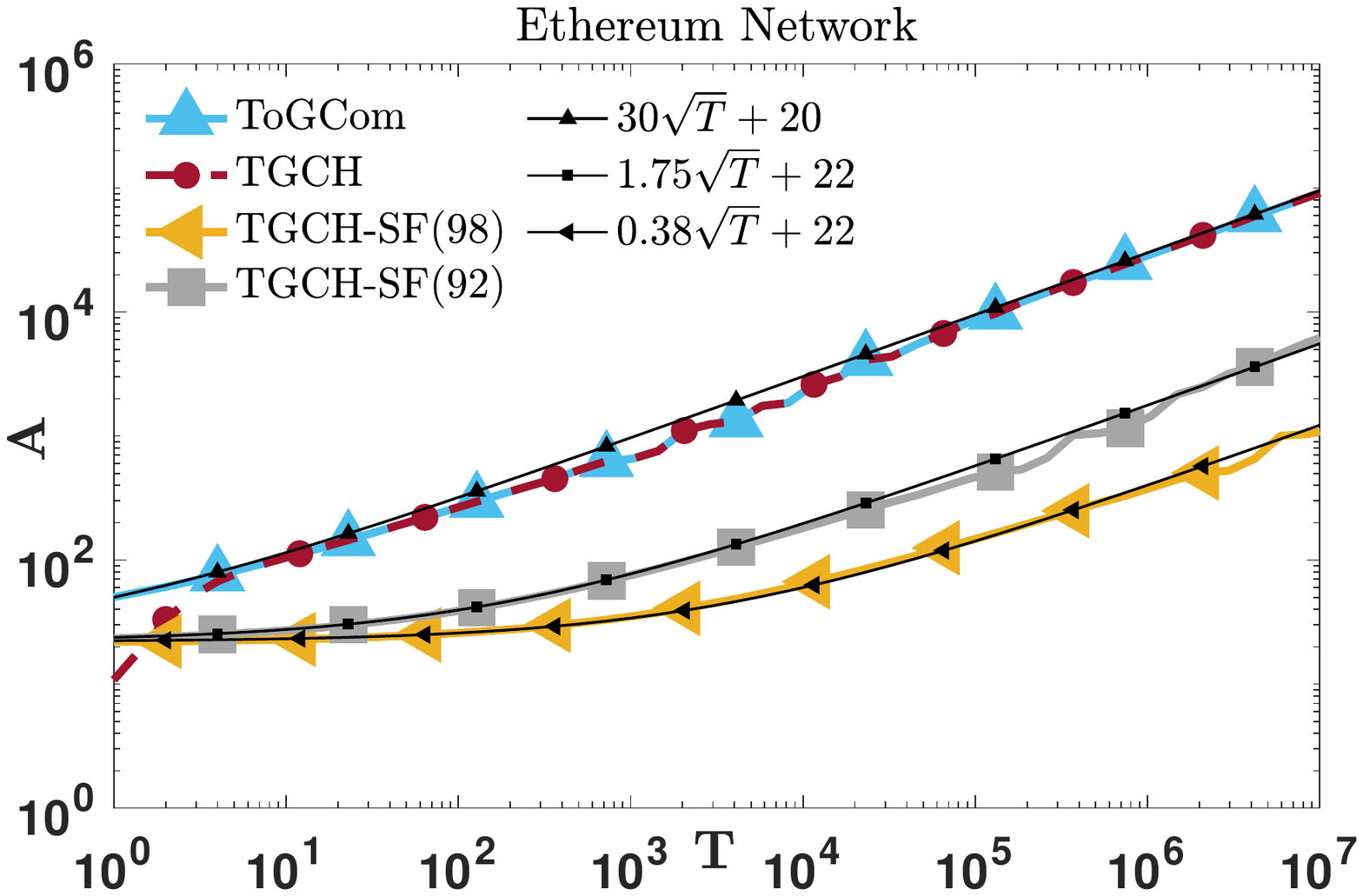} 
	\vspace{-20pt}
\end{subfigure}
\begin{subfigure}{0.42\textwidth}
	\includegraphics[trim = 1.5cm 8cm 2cm 8cm,  width=1\textwidth]{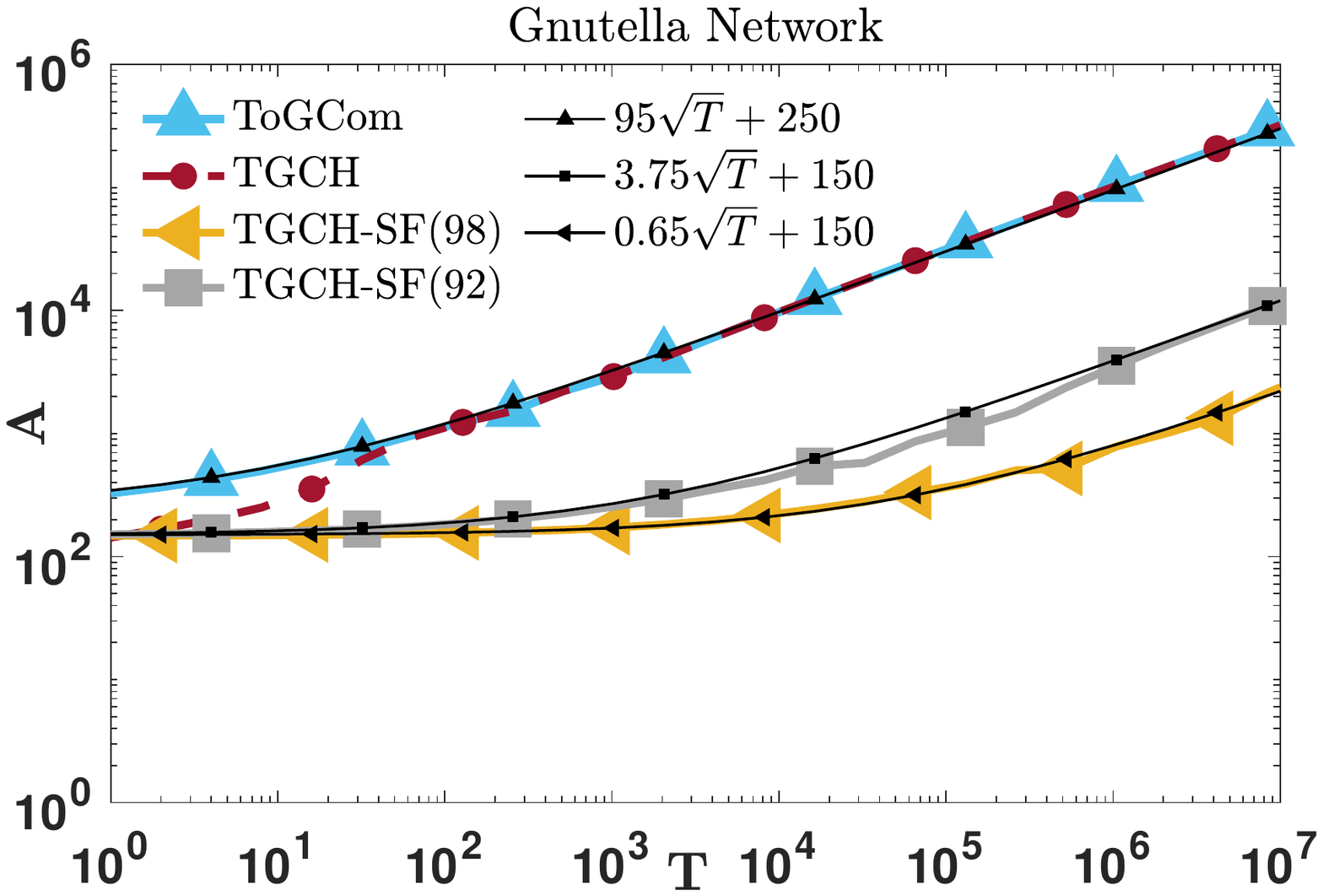} 
	\vspace{-20pt}
\end{subfigure}
\caption{Algorithmic cost versus adversarial cost for \tog and heuristics.}

\label{fig:Heuristic}
\end{figure*}

Next, we present heuristics to improve the performance of \tog. To determine effective heuristics, we focus on two separate costs to good IDs: entrance cost and purge cost. In studying these costs for the Bitcoin network in the absence of an attack, we find that the purge cost dominates and, therefore, we focus on reducing the purge frequency.

\smallskip

\noindent\textbf{Heuristic 1:} We use the symmetric difference to determine when to do a purge.  Specifically, for iteration $i$, if $|(\curIDs \cup \setIDs_{i-1}) - (\curIDs \cap \setIDs_{i-1})| \geq |\setIDs_{i-1}|/11$, then a purge is executed.  This ensures that the fraction of bad IDs can increase by no more than in our original specification.  Also, it decreases the purge frequency: for example, in the case when some ID joins and departs repeatedly.

\smallskip

\noindent\textbf{Heuristic 2:} We use the estimated good-ID join rate, obtained via \EstGoodJoin, to bound the maximum number of bad IDs that can have joined during an iteration.  This allows us to upper-bound the fraction of bad IDs in the system and to purge only when the population invariant is at risk.

\smallskip

\noindent\textbf{Heuristic 3:} Recent works have explored the possibility of identifying bad IDs based on the network topology~\cite{misra2016sybilexposer,gao2018sybilfuse}. In our experiments, we focus on SybilFuse~\cite{gao2018sybilfuse}, which has the probability of correctly classifying an ID as either good or bad as $0.92$ and $0.98$ based on the  empirical results from \cite{gao2018sybilfuse}, Section IV-B, last paragraph. We assume these values hold and use SybilFuse to diagnose whether a joining ID is good or bad; in the latter case, the ID is refused entry.

\medskip

We evaluate the performance of these heuristics against \tog. The experimental setup is the same as Section~\ref{section:empasym}. We define \defn{TGCH} to be \tog using both Heuristic 1 and Heuristic 2. We define \defn{TGCH-SF(92)} and \defn{TGCH-SF(98)} to be \tog using Heuristics 1 and 2, and also Heuristic 3, with the accuracy parameter of Heuristic 3 as 0.98 and 0.92.

Figure \ref{fig:Heuristic} 
illustrates our results. 
Note that TGCH-SF(92) and TGCH-SF(98) reduce costs significantly during adversarial attack, with improvements of up to three orders of magnitude during the most significant attack tested. Again, these improvements are consistent across $4$ different types of data sets.



\section{Related Work} \label{sec:related-work}

There is large body of literature on defending against the Sybil attack~\cite{douceur02sybil}; for example, see surveys~\cite{newsome:sybil,mohaisen:sybil,john:soft}, and additional work documenting real-world Sybil attacks~\cite{bitcoin-sybil,6503215,Yang:2011:USN:2068816.2068841}. To the best of our knowledge, \AlgB~\cite{pow-without} and \algGM~\cite{Gupta_Saia_Young_2019} are the first defenses where the algorithmic spend rate grows slowly with the adversarial spend rate.


\smallskip

\noindent{\bf Domain-Specific Defenses.} While \POW-based defenses work in general network settings, domain-specific results for mitigating the Sybil attack have been discovered.  In a wireless network with multiple communication channels, Sybil attacks can be mitigated via \emph{radio-resource testing} which relies on the inability of the adversary to listen to many channels simultaneously~\cite{monica:radio,gilbert:sybilcast,gilbert:who}.  However, this approach may fail if the adversary can monitor most or all of the channels. Furthermore, even in the absence of attack, radio-resource testing requires testing at fixed intervals. 

\smallskip
\noindent\textit{\bf Social Network Properties.}  Several results that leverage social networks for Sybil resistance \cite{yu:sybilguard,mohaisen:improving,wei:sybildefender}. However, social-network information may not be available in many settings. Another idea is to use network measurements to verify the uniqueness of IDs~\cite{sherr:veracity,liu:mason,Gil-RSS-15}, but these techniques rely on accurate measurements of latency, signal strength, or round-trip times, for example, and this may not always be possible. Containment strategies are explored in overlays~\cite{danezis:sybil,scheideler:shell}, but these results do not ensure a bound on the fraction of bad IDs.

\smallskip
\noindent\textit{\bf Proof of Work and Alternatives.}  As a choice for~\POW, computational puzzles provide certain advantages. First, verifying a solution is much easier than solving the puzzle itself. This places the burden of proof on devices that wish to participate in a protocol rather than on a verifier. In contrast, bandwidth-oriented schemes, such as~\cite{walfish2010ddos}, require verification of sufficient number of packets being received before any service is provided; this requires effort by the verifier that is proportional to the number of packets. 

A recent alternative to \POW is \defn{proof-of-stake (PoS)} where security relies on the adversary holding a minority stake in an abstract finite resource~\cite{abraham:blockchain}. When making a group decision, PoS weights each participant's vote using its share of a limited resource; for example, the amount of cryptocurrency held by the participant.  A well-known example is ALGORAND~\cite{GiladHMVZ17}, which employs PoS to form a committee.   A hybrid approach using both \POW and PoS has been proposed in the Ethereum system~\cite{ethereum-pos}.  


\section{Conclusion}\label{sec:future}

We have presented and analyzed \tog, which proposes a novel method for setting the entrance-puzzle difficulty, and for estimating the good-ID join rate. Additionally, we obtained a tight analysis of the asymmetric property. Our proposed heuristics further improve performance, and  experiments show that \tog decreases computational cost to the good IDs compared to other \POW-based Sybil defenses.

%
%

%
%
%

\newpage
\appendix

\section{Appendix}\label{sec:appendProofCost}

\smallskip

\subsection{Proof of Theorem~\ref{t:JoinEst}}\label{app:estimate}
We say that interval $\estDur$ \defn{touches} an epoch if there is a point in time belonging to both the interval $\estDur$ and the interval of time corresponding to the epoch; it does not necessarily mean that $\ell$ completely contains the epoch, or vice versa.

In this section, for any time {\boldmath{$t$}}, let {\boldmath{$\setIDs_t$}} and {\boldmath{$\setGood_t$}} denote the set of all IDs and set of good IDs, respectively, in the system at time $t$. 

\begin{lemma}\label{lem:interval-epochs}
An interval touches at most two epochs and cannot completely overlap any single epoch.
\end{lemma}

\begin{proof}
Assume that some interval starts at time $t_0$ and touches at least three epochs; we will derive a contradiction. This assumption implies that there is at least one epoch entirely contained within the interval.  Consider the first such epoch, and let it start at time $t_1\geq t_0$ and end at time $t_2> t_1$. Observe that:
 \begin{eqnarray*} 
 |\setIDs_{t_2}  -\setIDs_{t_1} |  &\geq& | \setGood_{t_2} - \setGood_{t_1}| \geq \frac{3}{4}| \setGood_{t_2}| \geq \left(\frac{5}{6}\right)\left(\frac{3}{4}\right)|\setIDs_{t_2}| = \frac{5}{8}|\setIDs_{t_2}| > \frac{3}{5} |\setIDs_{t_2}| 
\end{eqnarray*}
where step 2 holds by by the definition of an epoch and step 3 by the Population Invariant. But this is a contradiction since it implies that the interval must end before time $t_2$.\qed
\end{proof}

The following lemma considers any interval, $i$, of length $\estDur$, where there are $|\curIDs|$ IDs in the system at the end of the interval.  By the \sgoal, the number of good IDs in the system is always at least $(5/6)|\curIDs|$ .  Hence,  $\frac{|\curIDs|}{\estDur}$ is within constant factors of the good join rate during the interval. 

Thus, the next lemma shows that the estimate set by the algorithm in interval $i$, $\JoinEst_i (= \frac{|\curIDs|}{\estDur})$, is within constant factors of the good join rate during interval $i$.  (See Figure~\ref{fig:estimating-JG}).


\begin{lemma}\label{lem:bounded}
Consider any interval $i \geq 1$, and let epoch $j$ be the most recent epoch that the interval touches. Then:
\end{lemma}
$$\left(\frac{5}{6}\right) \AOneL \ATwoL \epochRate_j  \leq \JoinEst_i \leq 5 \AOneH \ATwoH \epochRate_j.$$
\begin{proof}

Let interval $i$ be of length $\estDur$, and assume there are $|\curIDs|$ IDs in the system at the end of the interval $i$.  Then $\JoinEst_i = \frac{|\curIDs|}{\estDur}$.  There are three cases.\smallskip

\noindent{\bf Case 1:} Interval $i$ touches only a single epoch where the epoch begins at $t_0$, the interval begins at $t_1$, and the interval ends at $t_2$. By assumption \Atwo, we have:
$$\ATwoL \epochRate_j \leq \frac{|\setGood_{t_2} - \setGood_{t_1}|}{t_2 - t_1} \leq   \ATwoH \epochRate_j$$

\noindent and by the Population Invariant and the specification of an interval: 
$$\left(\frac{2}{5}\right)|\setIDs_{t_2}| \leq \left(\frac{3}{5} - \frac{1}{6}\right)|\setIDs_{t_2} | \leq |\setGood_{t_2} - \setGood_{t_1}| \leq \left(\frac{3}{5}\right)|\setIDs_{t_2} |. $$

\noindent Therefore, we have:
$$ \left(\frac{5}{3}\right) \ATwoL \epochRate_j \leq \frac{ |\setIDs_{t_2} |}{t_2 - t_1} \leq \left(\frac{5}{2}\right) \ATwoH \epochRate_j $$
\noindent{\bf Case 2:}  The interval touches epochs $i-1$ and $i$ and there are at least $2$ good ID join events in each epoch. Let epoch $i-1$ start at time $t_0$ and end at $t_2$ (and so epoch $i$ starts at $t_2$), and let the interval start at time $t_1 \geq t_0$ and end at $t_3 \geq t_2$. By Assumptions \Aone and A2, we have:
$$ \AOneL \ATwoL \epochRate_j \leq \frac{|\setGood_{t_3} - \setGood_{t_1}|}{t_3 - t_1} \leq  \AOneH \ATwoH \epochRate_j $$

\noindent and by the Population Invariant and the specification of an interval:
$$\left(\frac{2}{5}\right)|\setIDs_{t_3}| \leq \left(\frac{3}{5} - \frac{1}{6}\right)|\setIDs_{t_3} | \leq |\setGood_{t_3} - \setGood_{t_1}| \leq \left(\frac{3}{5}\right)|\setIDs_{t_3} |. $$

\noindent Therefore, we have:
$$\left(\frac{5}{3}\right)\hspace{-2pt}\AOneL \ATwoL \epochRate_j  \hspace{-2pt}\leq \hspace{-2pt}\frac{ |\setIDs_{t_3} |}{t_3 - t_1} \hspace{-2pt}\leq\hspace{-2pt}  \left(\frac{5}{2} \right)\hspace{-2pt} \AOneH \ATwoH \epochRate_j  $$
\noindent which we refer to as the Case-2 Equation.

\smallskip
\noindent{\bf Case 3:} The interval touches epochs $j-1$ and $j$, and w.l.o.g. we have a single good join event in the portion of epoch $j-1$ that overlaps the interval; denote the length of this overlap by $\lambda'>0$.  As with Case 2, let epoch $j-1$ start at time $t_0$ and end at $t_2$ (and so epoch $j$ starts at $t_2$), and let the interval start at time $t_1 \geq t_0$ and end at $t_3 \geq t_2$. 

Observe that the single good join event in epoch $j-1$ increases the numerator of the bounded quantity in the  Case-2 Equation by $1$, and so twice the upper bound in Case 2 suffices here.  The denominator increases by $\lambda'$, where $\lambda' \leq t_3-t_1$, so half the lower bound in Case 2 suffices. This implies:
$$\left(\frac{5}{6}\right) \AOneL \ATwoL \epochRate_j  \leq \frac{ |\setIDs_{t_3} |}{t_3 - t_1} \leq  5 \AOneH \ATwoH \epochRate_j   $$
\qed
\end{proof}

\begin{lemma}\label{lem:four}
For any epoch $j\geq 1$, and any iteration $i \geq 1$ that epoch $j$ touches:
\end{lemma}
$$\AOneL \iJRate_{i} \leq \epochRate_{j} \leq \AOneH \iJRate_{i}$$
\begin{proof}
\noindent By Lemma~\ref{lem:interval-epochs}, iteration $j$ touches at most two epochs, say epochs $j-1$ and $j$. Thus, by definition of $\iJRate_{i}$, for some $\lambda_1, \lambda_2 \geq 0$, $\lambda_1+\lambda_2 = 1$:
$$\iJRate_{i} = \lambda_{1} \epochRate_{j-1} + \lambda_{2} \epochRate_{j} \leq  \lambda_{1} \left(\frac{\epochRate_{j}}{ \AOneL}\right) + \lambda_{2} \epochRate_{j} \leq  \frac{\epochRate_{j}}{ \AOneL}$$
The second step above holds by assumption \Aone.  The last step holds since $\lambda_1 + \lambda_{2} = 1$ and  $\AOneL\leq 1$.  A similar derivation yields that $\iJRate_{i}  \geq  \frac{\epochRate_{j}}{\AOneH}$. Together this implies that:
$ \AOneL \iJRate_{i} \leq \epochRate_{j} \leq \AOneH \iJRate_{i}.$
\qed
\end{proof}

\noindent We can now prove Theorem~\ref{t:JoinEst}, restate below.

\begin{customthm}{1}
For any iteration $i\geq 1$, \EstGoodJoin provides a constant-factor estimate of  the good-ID join rate: 
$$ \smallconstant \iJRate_{i} \leq \hspace{-1pt}\JoinEst_i \hspace{-3pt} \leq \bigconstant \iJRate_{i}$$
where $\smallconstant = \left(\frac{5}{6}\right) \cfrac{\AOneL^2 \ATwoL}{\AOneH} $ and $\bigconstant = \cfrac{5\,\AOneH^2 \ATwoH}{\AOneL}$.
\end{customthm}

\begin{proof}
The estimate $\JoinEst_i$ used in iteration $i$ corresponds to the most recent interval that completed before iteration $i$ started. Let epoch $j$ be the most recent epoch that touches this interval. By Lemma~\ref{lem:bounded}:
$$\left(\frac{5}{6}\right) \AOneL \ATwoL \epochRate_j \leq\JoinEst_i \leq 5\,\AOneH \ATwoH \epochRate_j.$$

Epoch $j$ may end prior to the start of iteration $i$; that is, epoch $j$ may not necessarily touch iteration $i$. In this case, note that by Lemma~\ref{lem:interval-epochs}, the current interval touches epoch $j+1$ and must end before epoch $j+1$ ends. This fact, along with the observation that the current interval touches iteration $i$, implies that epoch $j+1$  touches iteration $i$. 

By the above, we know that either epoch $j+1$ or epoch $j$ touches iteration $i$. Lemma~\ref{lem:four} implies:
$$ \AOneL \iJRate_{i} \leq \epochRate_{j} \leq \AOneH \iJRate_{i}$$ 
or
$$ \AOneL \iJRate_{i} \leq \epochRate_{j+1} \leq \AOneH \iJRate_{i}.$$

\noindent To employ our top-most equation, we use \Aone to derive:
$$ \epochRate_{j+1}/\AOneH \leq \epochRate_{j} \leq  \epochRate_{j+1}/\AOneL  $$
\noindent and then plug into our top-most equation, we have:
$$\left(\frac{5}{6}\right) \frac{\AOneL^2 \ATwoL}{\AOneH}  \iJRate_{i} \leq\hspace{-1pt}\JoinEst_i \hspace{-3pt}\leq \frac{5\,\AOneH^2 \ATwoH}{\AOneL}  \iJRate_{i}$$
\qed
\end{proof}

\smallskip

\subsection{Proof of Theorem~\ref{thm:new-main-upper}}\label{sec:appendasym}
We make use of the following algebraic fact that follows from the Cauchy-Schwartz inequality.

\begin{lemma}\label{l:cs2}
Let $n$ be a positive number, and for all $1 \leq i \leq n$, $s_i \geq 0$ and let $S = \sum_{i=1}^n s_i$.  Then
$$ \sum_{i=1}^n s_i^2 \geq S^2/n$$
\end{lemma}
\begin{proof}

Let $u$ be a vector of length $n$ with for all $1 \leq i \leq n$, $u[i] = s_i$, and let $v$ be a vector of length $n$ with, for all $1 \leq i \leq n$, $v[i] = 1$.  Then by Cauchy-Schwartz:
\begin{eqnarray*}
	|\langle u, v \rangle|^2 & \leq & \langle u,u \rangle \cdot \langle v, v \rangle \\
	S^2  & \leq & \left(\sum_{i=1}^n s_i^2 \right) \cdot n
\end{eqnarray*}
Rearranging completes the proof.\qed
\end{proof}

\noindent Let {\boldmath{$\advCost_i$}} denote the cost to the adversary over iteration $i$ divided by the length of iteration $i$. Let {\boldmath{$\joinRateBad_{i}$}} be the join rate of bad IDs during iteration $i$. Recall that $\bigconstant \iJRate_{i} = \frac{ 5 \AOneH^2 \ATwoH}{\AOneL}$.


\begin{lemma}\label{l:joinBad}
For any iteration $i>1$, 
$$\joinRateBad_i \leq d_1 \sqrt{T_i (J^G_i+1)}$$ where $d_1 = \sqrt{2\bigconstant}$.
\end{lemma}

\begin{proof}

For simplicity, we normalize time units so that $\ell_i=1$.  Partition iteration $i$ from left to right into sub-iterations, all of length $1/\JoinEst_i$, except the last, which is of length at most $1/\JoinEst_i$.  We lower bound the cost paid by the adversary for joins by pessimistically assuming that only bad IDs are counted when computing entrance costs. For  $1 \leq x \leq \lceil \JoinEst_i  \rceil$, let $j_x$ be the total number of bad IDs that join in sub-iteration $x$.  Since $\sum_{y=1}^{j_x} y = (j_x + 1)j_x/2 \geq (j_x)^2/2$,  the total entrance cost paid by bad IDs is at least $(1/2) \sum_{x = 1}^{\lceil \JoinEst_i  \rceil} (j_x)^2$. 
 
\noindent zSince $\sum_{x = 1}^{\lceil \JoinEst_i  \rceil} j_x = \joinRateBad_i$,   by applying Lemma~\ref{l:cs2}, we have:
$$T_i \geq \frac{1}{2} \sum_{x = 1}^{\lceil \JoinEst_i  \rceil} (j_x)^2 \geq \frac{(\joinRateBad_i)^2}{2 (\lceil \JoinEst_i  \rceil)} $$	
Cross-multiplying and taking the square root, we get:
\begin{eqnarray*}
	\joinRateBad_i & \leq & \sqrt{2 T_i \lceil \JoinEst_i  \rceil} \leq \sqrt{2 T_i (\JoinEst_i +1)} \leq \sqrt{2  T_i \left(\bigconstant \iJRate_{i} + 1\right)} 
\end{eqnarray*}
where the second step follows from noting that $\lceil x \rceil \leq x +1$ for all $x$, and the final step follows from Lemma~\ref{t:JoinEst} which states that: 
$$\JoinEst_i  \leq  \frac{ 5 \AOneH^2 \ATwoH}{\AOneL}  \iJRate_{i} = \bigconstant \iJRate_{i} $$
\noindent which yields the lemma statement. \qed
\end{proof}

\begin{lemma} \label{l:algCost}
		Let $\mathcal{A}_i$ be the average spend rate for the algorithm in any iteration $i>1$.  Then, $\mathcal{A}_i \leq \frac{d_2|S_{i-1}|}{\ell_i}$, where $d_2 = \left( \frac{12}{11} + \frac{\AOneH \ATwoH}{11\smallconstant} \right)$.
\end{lemma}
\begin{proof}
For simplicity, we first normalize time units so that $\ell_i=1$.  Partition iteration $i$ from left to right into sub-iterations, all of length $1/\joinRate_i$, except the last, which is of length at most $1/\joinRate_i$.  

The spend rate for the algorithm due to purge costs is  $|S_{i-1}|$.  

For entrance costs, note the following two facts.  First, by Assumptions A1 and A2, there are at most $\AOneH \ATwoH $ good IDs in any sub-iteration (note that a sub-iteration might span two epochs). Second,  the entrance cost for any good ID is $1$ plus the number of join events over the past $1/\JoinEst_i$ seconds. By Theorem~\ref{t:JoinEst}, $1/\JoinEst_i \leq 1/(\smallconstant\joinRate_i)$, and so the entrance cost is at most $1$ plus the number of join events over the past $1/\smallconstant$ sub-iterations; that is $1 + \AOneH \ATwoH/\smallconstant$.

By these two facts, and by the fact that there are at most $|S_{i-1}|/11$ join events in an iteration, the total entrance costs paid by good IDs in iteration $i$ is at most:
$$\left(1 + \frac{\AOneH \ATwoH}{\smallconstant}\right)\left(\frac{|S_{i-1}|}{11}\right).$$  

Adding the bounds for both entrance and purge costs and dividing by the value of $\ell_i$ yields a total cost of at most:

$$ \frac{|S_{i-1}|}{\ell_i}\left( \frac{12}{11} + \frac{\AOneH \ATwoH}{11\smallconstant} \right)$$
\qed
\end{proof}

Let {\boldmath{$\Iters$}} be any subset of iterations that for integers $x$ and $y$, $1 \leq x \leq y$, contains every iteration with index between $x$ and $y$ inclusive.  Let {\boldmath{$\delta(\Iters)$}} be $|S_{x} - S_y|$; and let $\Delta(\Iters)$ be  $\delta(\Iters)$ divided by the length of $\Iters$. Let {\boldmath{$\mathcal{A}_{\Iters}$}} and {\boldmath{$T_{\Iters}$}} be the algorithmic and adversarial spend rates over $\Iters$; and let {\boldmath{$\iJRate_{\Iters}$}} be the good join rate over all of $\Iters$.  Then we have the following lemma.

\begin{lemma}\label{l:cost}
   	For any subset of contiguous iterations, $\Iters$,which starts after iteration $1$, the algorithmic spending rate over $\Iters$ is at most:
$$11d_2  \left(  2\Delta(\Iters)  +  d_1\sqrt{2T_{\Iters} (\bigconstant \iJRate_{\Iters} + 1)}  + \iJRate_{\Iters}\right).$$
\end{lemma}
\begin{proof}
\noindent Let {\boldmath{$\depRateTot_{i}$}} be the departure rate of both good and bad IDs during iteration $i$.  By Lemma~\ref{l:algCost} and the assumptions for when a purge occurs, we have that:
\begin{eqnarray*}
	\sum_{i \in \Iters} \mathcal{A}_i \ell_i & \leq &d_2\sum_{i \in \Iters} |S_{i-1}| \\
	&\leq & 11d_2\sum_{i \in \Iters} (\depRateTot_i + \joinRateBad_i + \iJRate_{i}) \ell_i \\
	& \leq & 11d_2\Bigg( 2 \delta(\Iters) +  \sqrt{ \sum_{i \in \Iters} 2T_i \ell_i  \sum_{i \in \Iters} (\bigconstant \iJRate_i + 1) \ell_i } + \sum_{i \in \Iters} \iJRate_{i} \ell_i\Bigg) 
\end{eqnarray*}

The second line in the above follows from the fact that every ID that departs must have departed from the set of IDs in the system at the start of $\Iters$ or else must have been an ID that joined during $\Iters$.  The last line follows from the Lemma~\ref{l:joinBad} bound on $\joinRateBad_i$, and noting that $\ell_i = \sqrt{\ell^2_i}$.  Finally, the last line follows from Cauchy-Schwarz Inequality.

\smallskip

Dividing both sides of the above inequality by $\sum_{i \in \Iters} \ell_i$ and recalling that $d_1 = \sqrt{2\bigconstant}$, completes the proof. \qed
\end{proof}
 
\noindent The above result is more general than we need, but \textbf{Theorem~\ref{thm:new-main-upper} follows from Lemma~\ref{l:cost}} by  noting that $\Delta(\Iters)=0$ when $\Iters$ is all iterations, since the system is initially empty. We state an interesting corollary implied by Theorem \ref{thm:new-main-upper}.

\begin{corollary}\label{cor:main-upper}
For $\alpha \leq 1/18$, with error probability polynomially small in $n_0$ over the system lifetime, and for any subset of contiguous iterations $\mathcal{I}$, \tog has an algorithmic spending rate of $O\left(\sqrt{\advAveCost_{\mathcal{I}}\,(\joinRate_{\Iters}+1)} + \Delta(\Iters)  + \iJRate_{\Iters}\right)$. 
\end{corollary}

\noindent This shows that the spending rate for the algorithm remains small, even when focusing on just a subset of iterations.  To understand why this is important, consider a long-lived system which suffers a single, significant attack for a small number of iterations, after which there are no more attacks. The cost of any defense may be small when amortized over the lifetime of the system, but this does not give a useful guarantee on performance during the time of attack.

\subsection{\algGM with close Join Events}\label{app:GMComFailure}

\begin{wrapfigure}[14]{r}{0.5\textwidth}
    \centering
    \includegraphics[trim = 1.8cm 7.5cm 1.8cm 10cm, width = 0.48\textwidth]{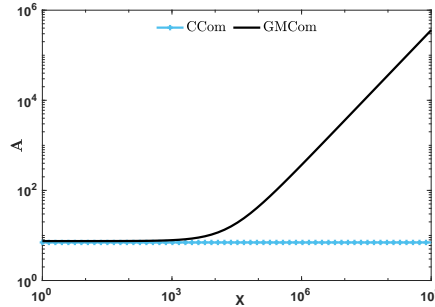}
    \vspace{-5pt}
    \caption{\small Comparison of algorithmic spend rate for \AlgB~and \algGM as join event grows closer to beginning of an epoch.}
    \label{fig:gmcomFail}
\end{wrapfigure}
   
\medskip   

\textbf{Experimetal Setup.} We simulate \AlgB and \algGM over a system that always consists of $10,000$ good IDs. During the lifetime of the system, new good IDs join/depart at a constant rate of one ID per time step. The departing ID is chosen uniformly at random from the set of IDs that joined before the current iteration. We run for 2 iterations and then have a single join event in the third iteration.  This last join event occurs arbitrarily close to second to last join event which ended the previous iteration.

Let the time between the second to last and the last join event be $1/X$, for some value $X >0$.   Thus, as $X$ increases, these two last join events become closer together. For $X \in \{2^0,2^1,...,2^{30}\}$, we compute the spend rate for \AlgB and \algGM. 

\smallskip

\noindent Figure \ref{fig:gmcomFail} illustrates our  results. As can be seen, as $X$ increases, the algorithmic spend rate, $\mathcal{A}$, increases linearly for \algGM, whereas for \AlgB it remains constant.  This increase in the spend rate for \algGM occurs solely because of the increasing entrance cost for the very last good join event.

\subsection{Plots from Experiments of Section~\ref{s:JandLAssum}}\label{app:plots}
\smallskip

\noindent In Figure~\ref{fig:BitcoinAssumptions}, we present the plots from our experiments testing Assumptions \Aone and \Atwo. These are used to derive the values presented in Table~\ref{tab:assumptions}.

\begin{figure*}
\captionsetup[subfigure]{labelformat=empty}
\centering
\begin{subfigure}{0.45\textwidth}
	\centering
	\includegraphics[trim = 1.25cm 7.5cm 1.25cm 7cm, width = 0.8\textwidth]{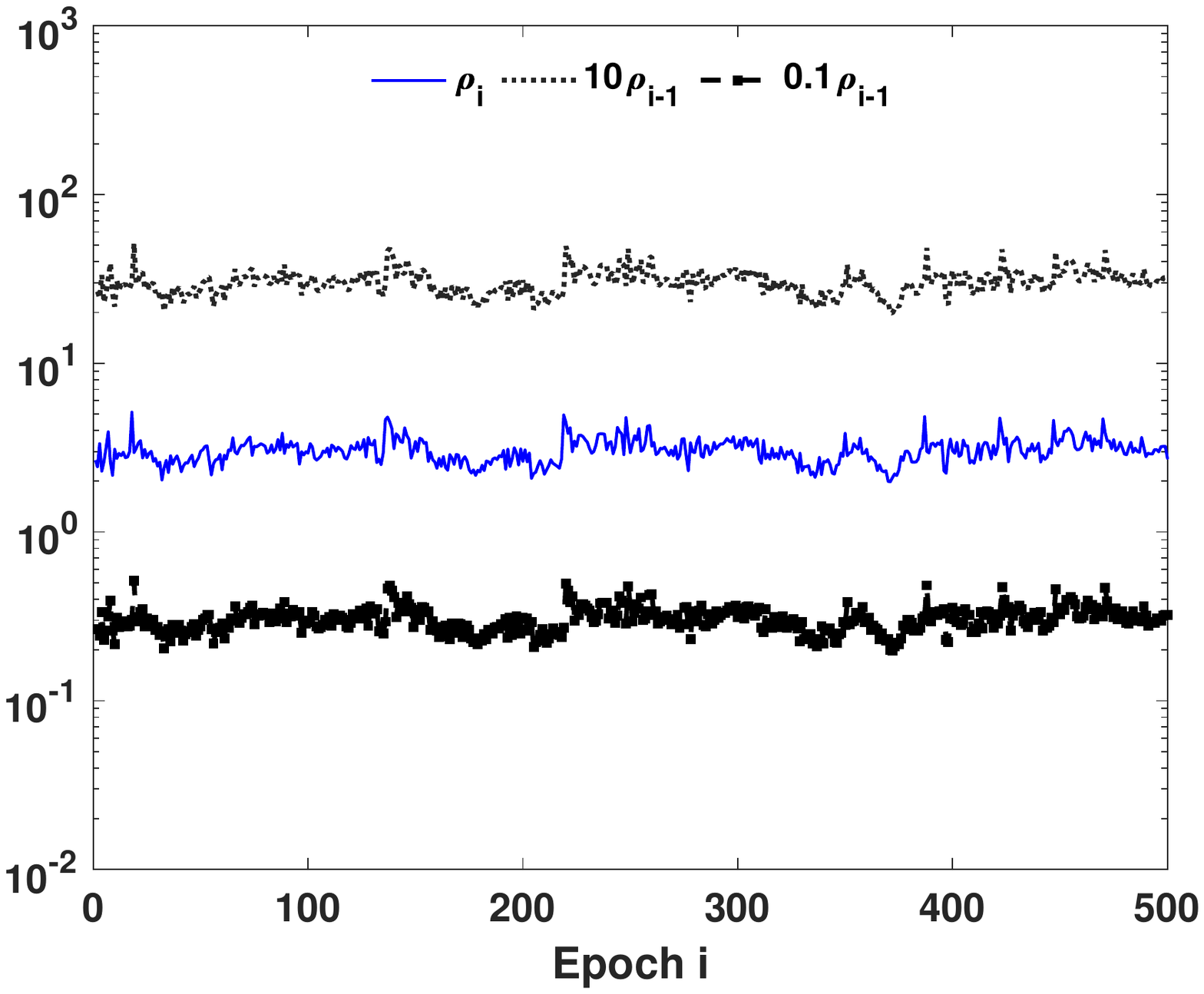} 
	\caption{(a)}
\end{subfigure}
\begin{subfigure}{0.45\textwidth}
	\includegraphics[trim = 1.25cm 7.5cm 1.25cm 7cm, width = 0.8\textwidth]{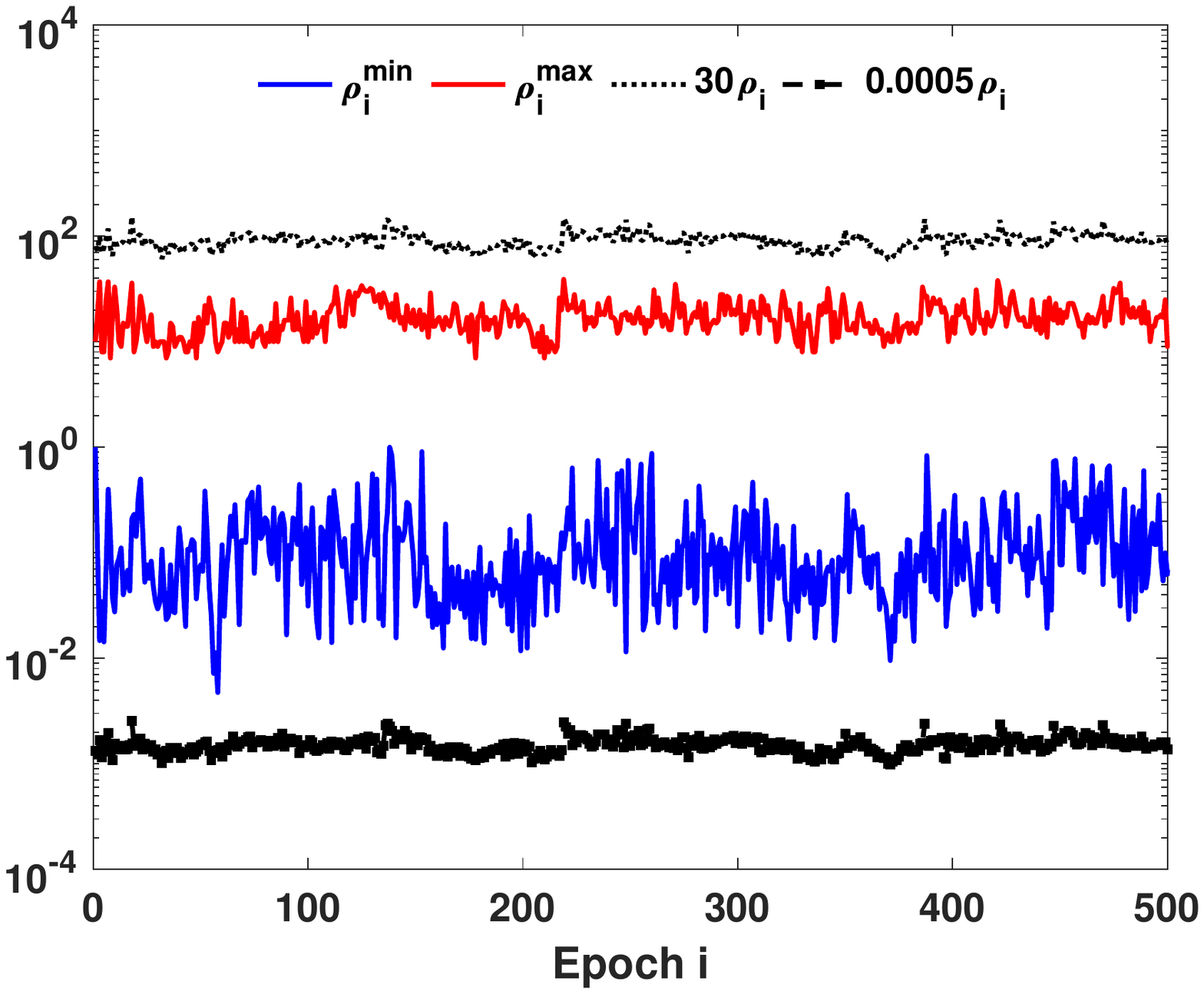}  
	\caption{(b)}
\end{subfigure}
\begin{subfigure}{0.45\textwidth}
	\centering
	\includegraphics[trim = 1.25cm 7.5cm 1.25cm 7cm, width = 0.8\textwidth]{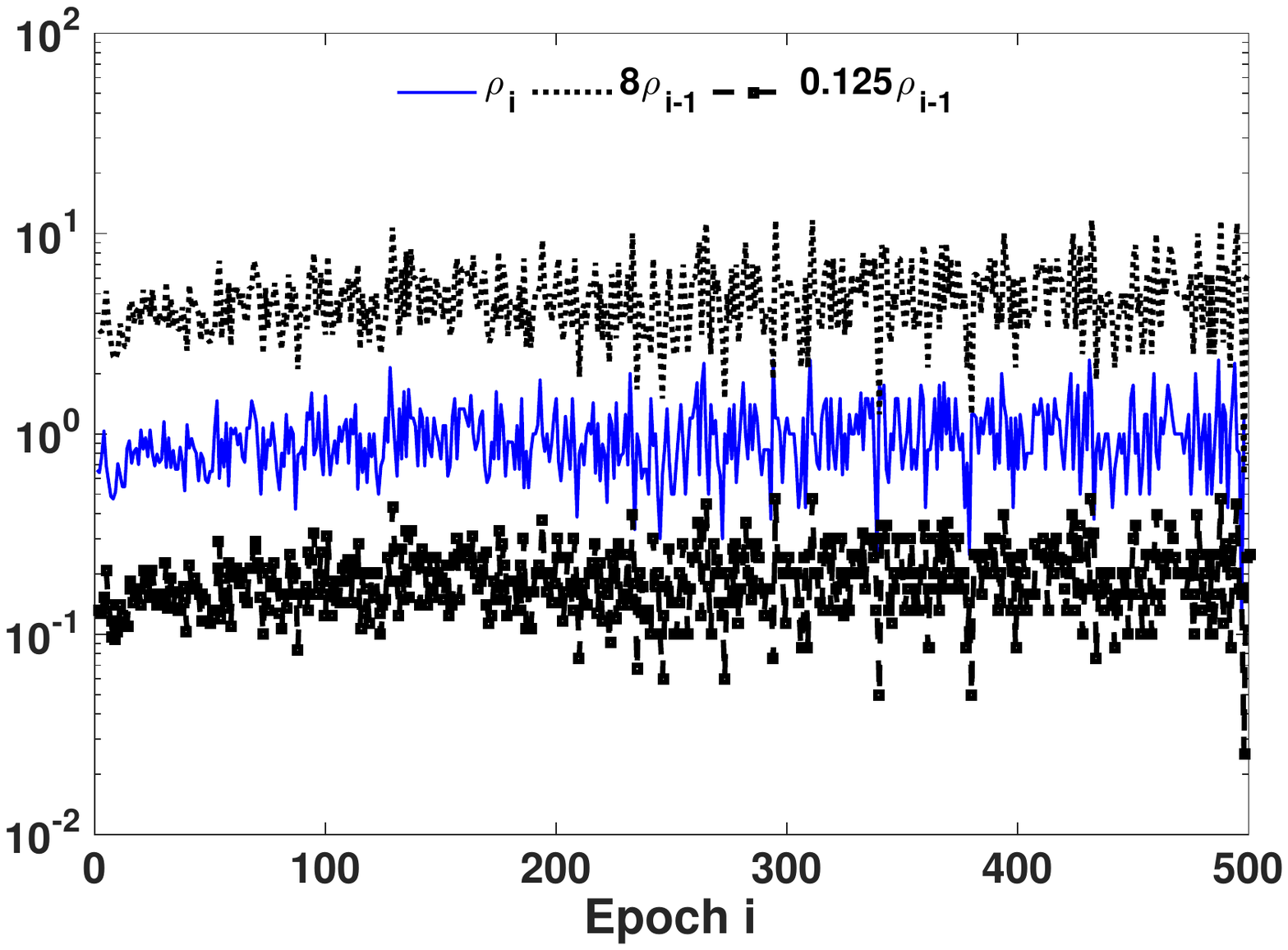} 
	\caption{(c)}
\end{subfigure}
\begin{subfigure}{0.45\textwidth}
	\includegraphics[trim = 1.25cm 7.5cm 1.25cm 7cm, width = 0.8\textwidth]{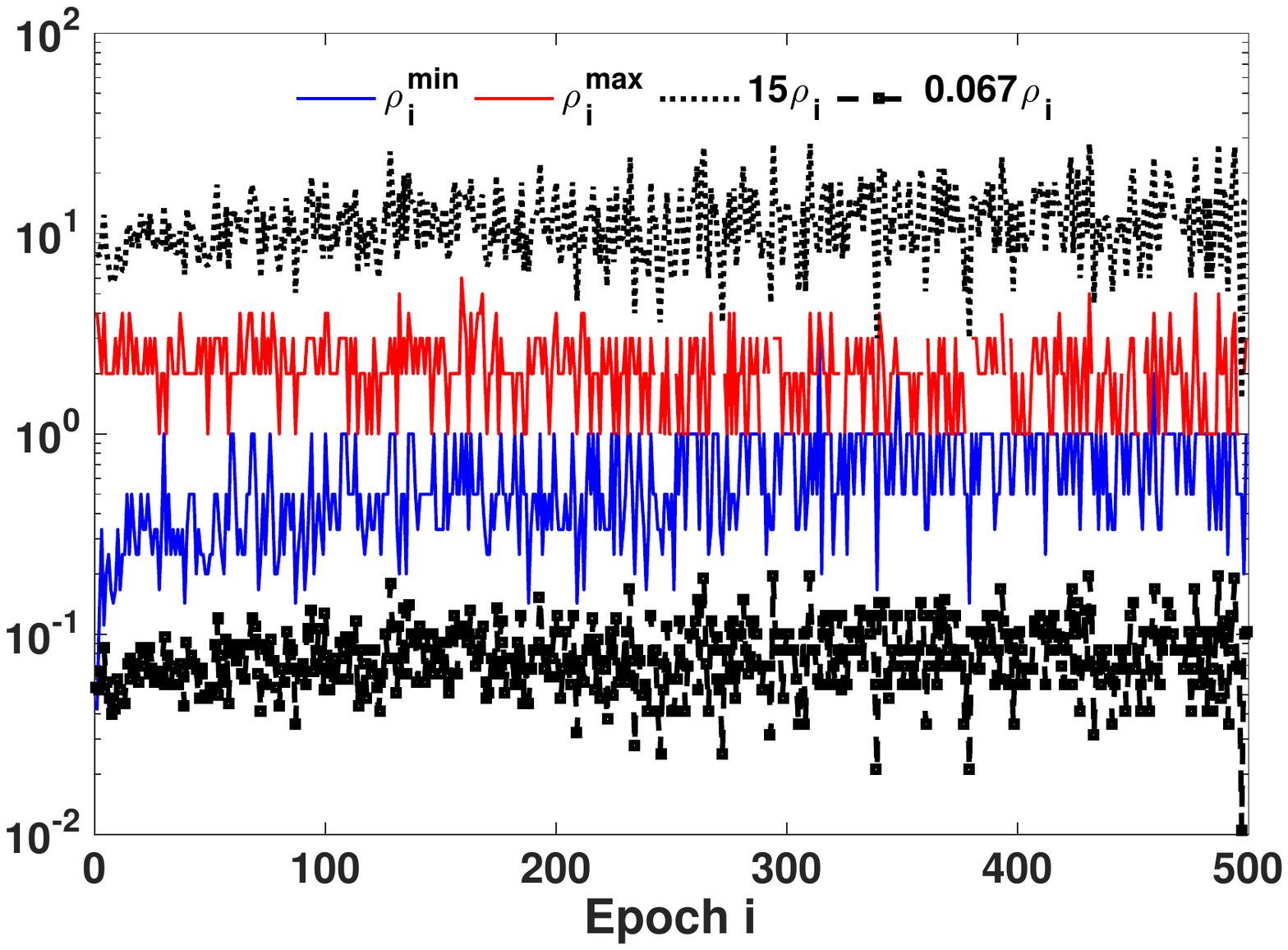}  
	\caption{(d)}
\end{subfigure}
\begin{subfigure}{0.45\textwidth}
	\centering
	\includegraphics[trim = 1.25cm 7.5cm 1.25cm 7cm, width = 0.8\textwidth]{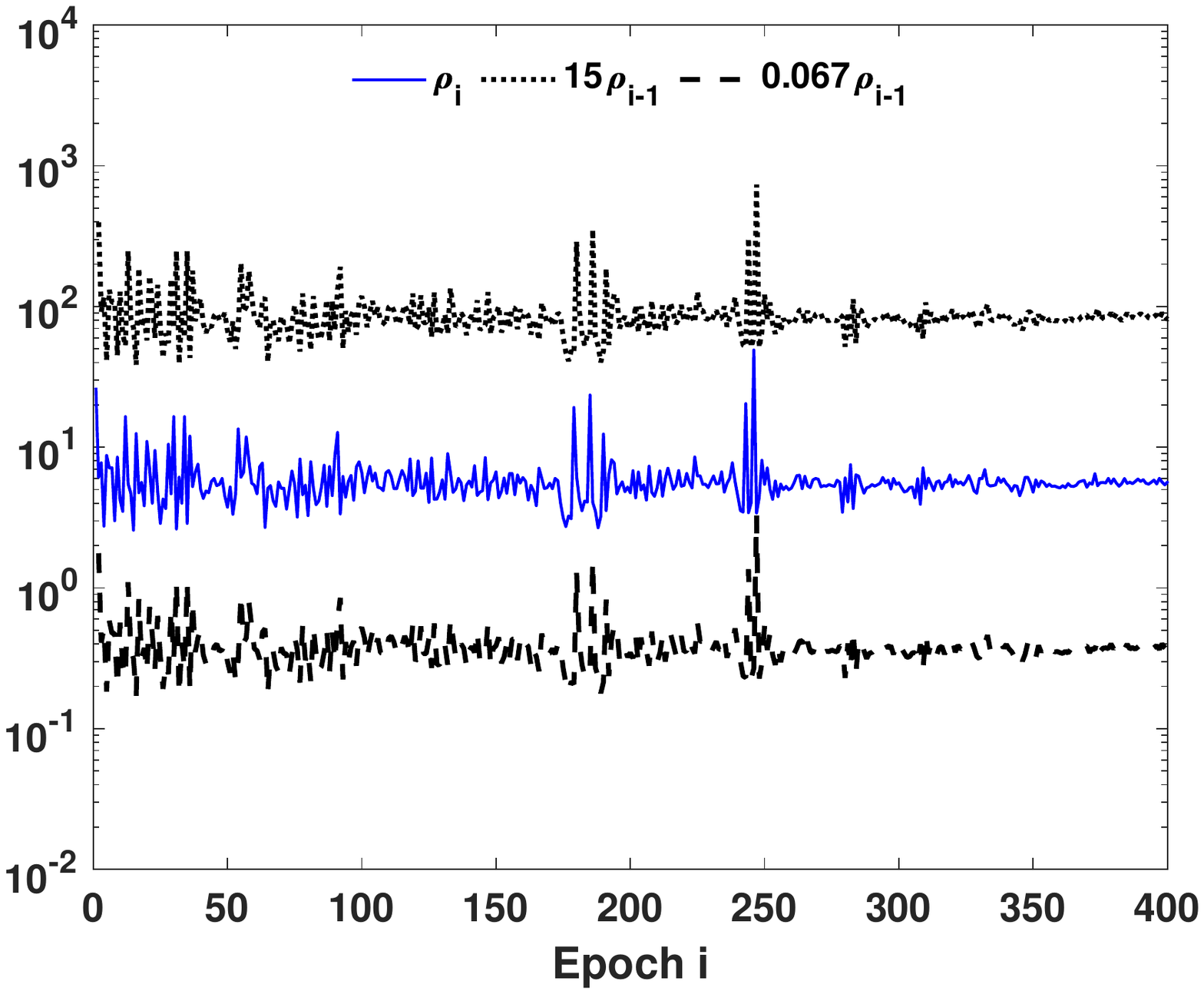} 
	\caption{(e)}
\end{subfigure}
\begin{subfigure}{0.45\textwidth}
	\includegraphics[trim = 1.25cm 7.5cm 1.25cm 7cm, width = 0.8\textwidth]{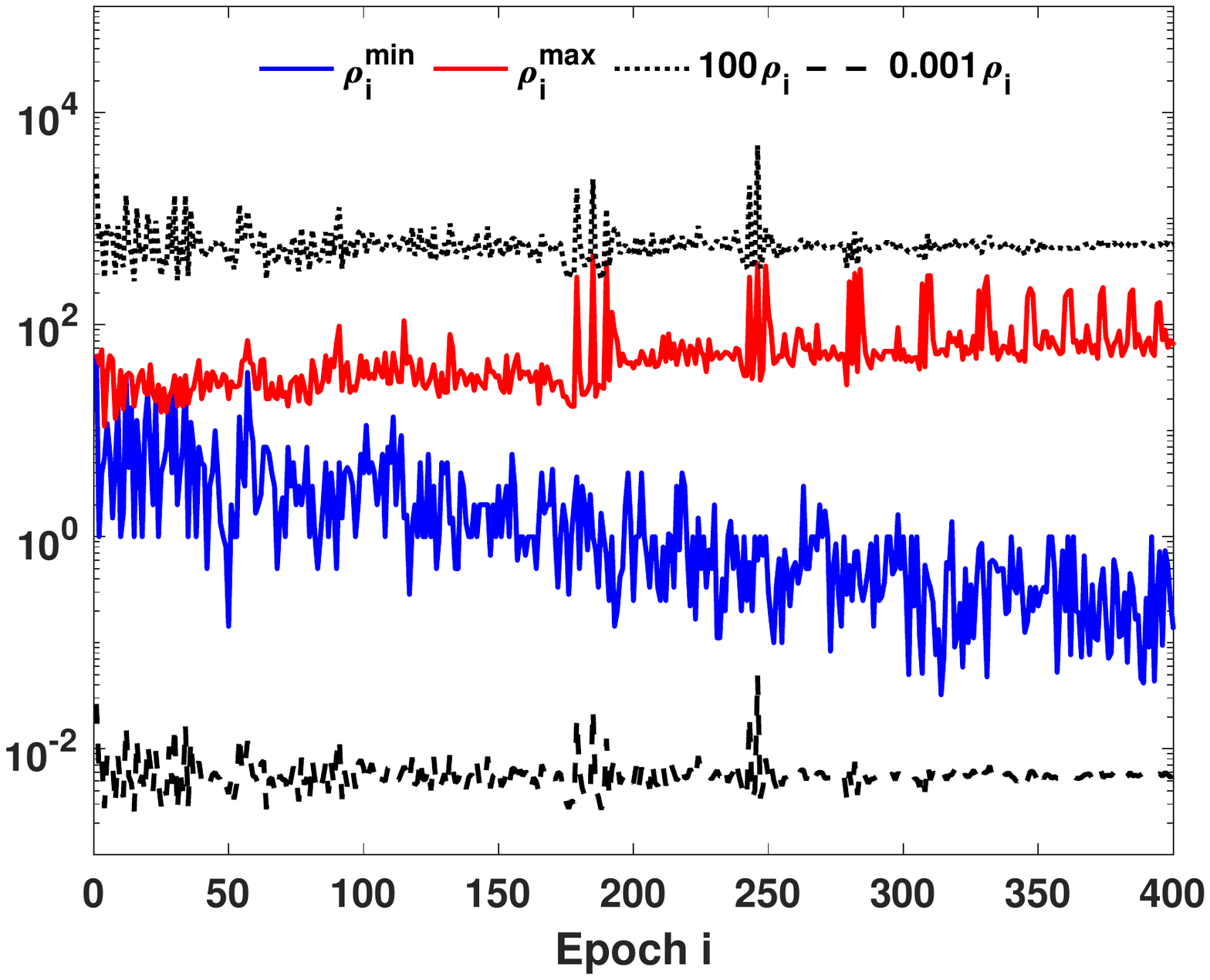}  
	\caption{(f)}
\end{subfigure}
\begin{subfigure}{0.45\textwidth}
	\centering
	\includegraphics[trim = 1.25cm 7.5cm 1.25cm 7cm, width = 0.8\textwidth]{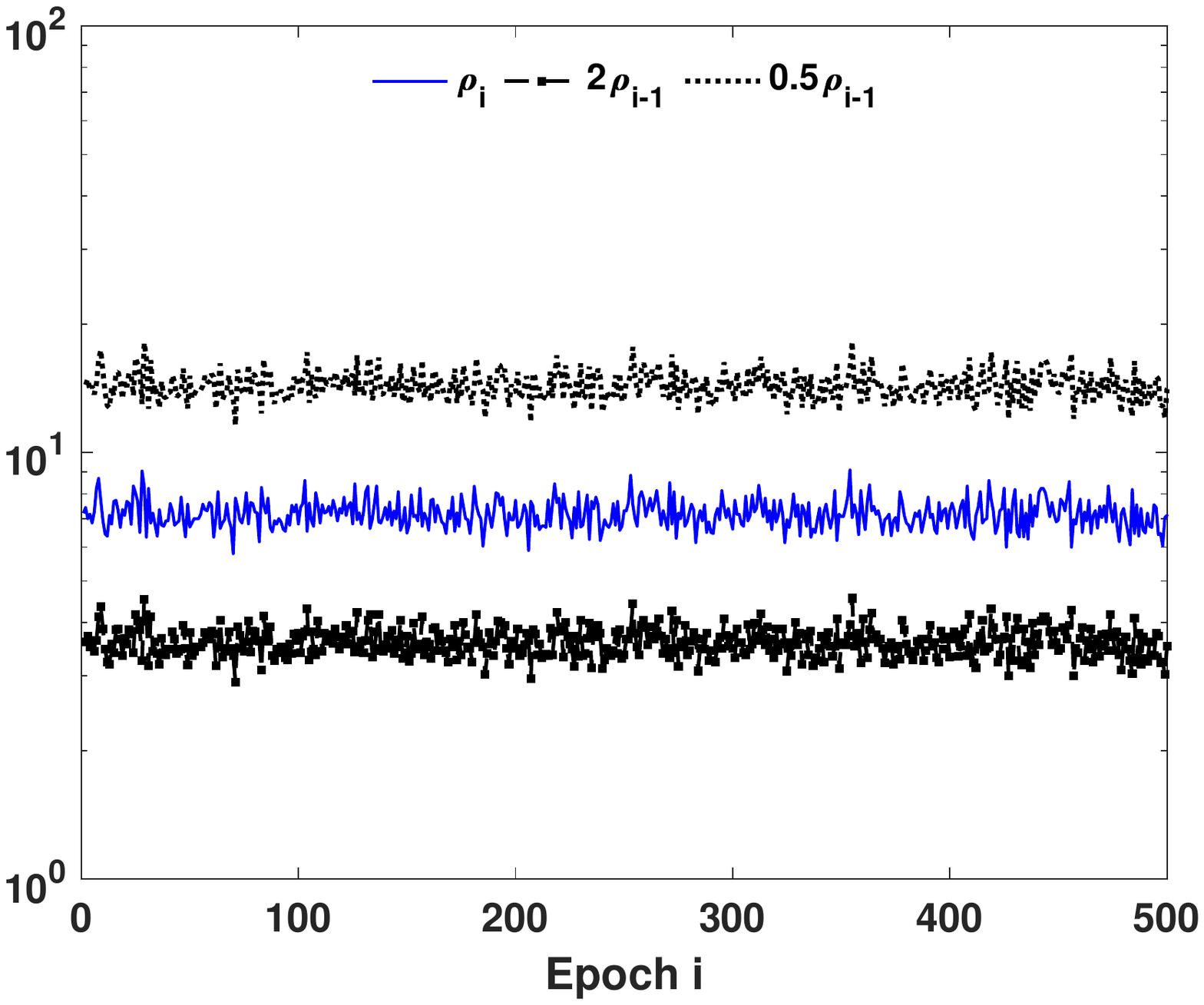} 
	\caption{(g)}
\end{subfigure}
\begin{subfigure}{0.45\textwidth}
	\includegraphics[trim = 1.25cm 7.5cm 1.25cm 7cm, width = 0.8\textwidth]{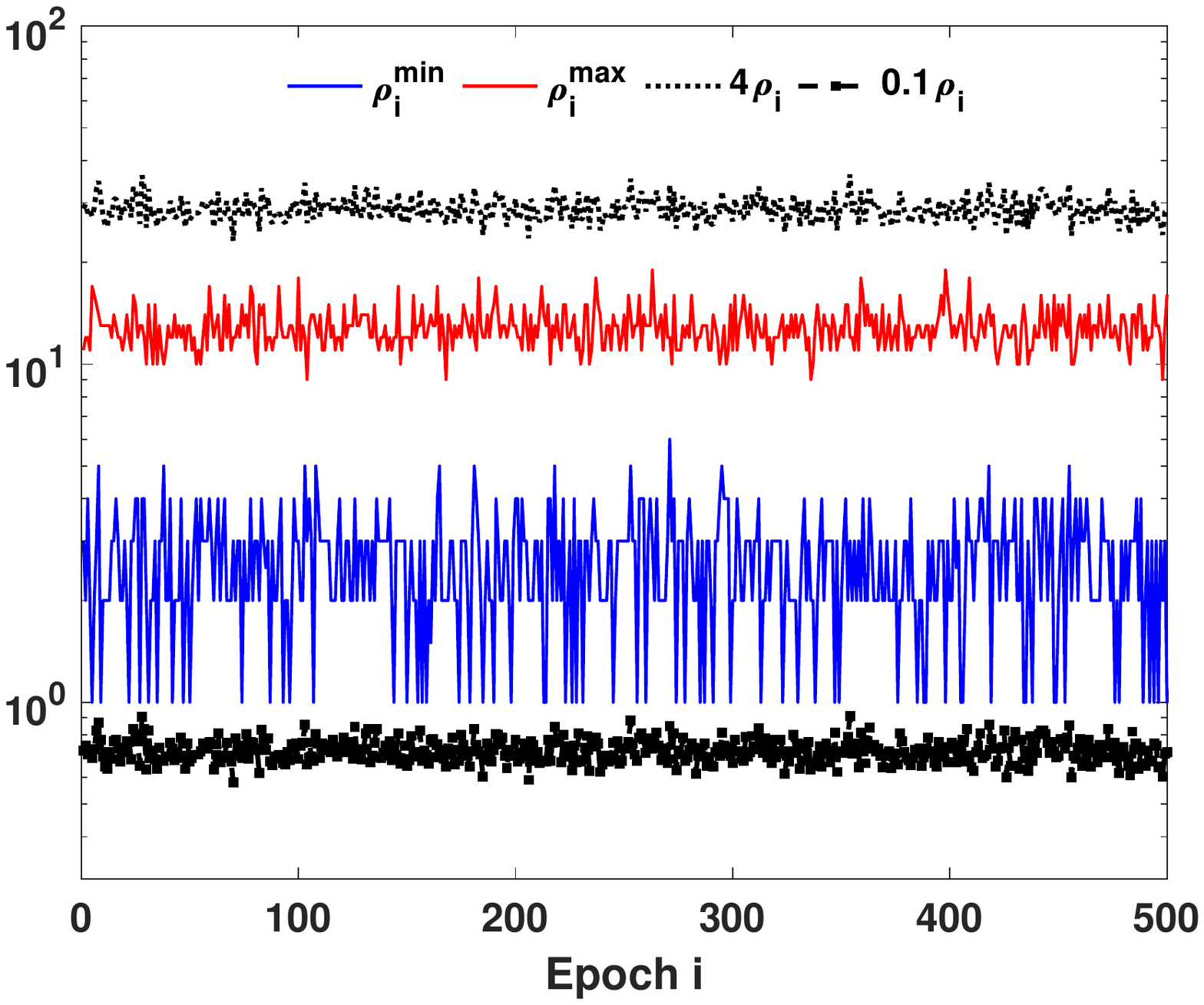}  
	\caption{(h)}
\end{subfigure}

\caption{Testing Assumptions \Aone and \Atwo for (a) \& (b) Bitcoin Network, (c) \& (d) BitTorrent Redhat, (e) \& (f) Ethereum Network, and (g) \& (h) Gnutella Network.}
\label{fig:BitcoinAssumptions}
\vspace{-5pt}
\end{figure*}

\end{document}